\documentclass[10pt, final]{IEEEtran}
\usepackage[dvips]{graphicx}
\usepackage{amssymb,amsmath}
\usepackage{subfig}
\usepackage{colortbl}
\usepackage{epsfig}
\usepackage{algorithm2e}

\IEEEoverridecommandlockouts
\newtheorem{theorem}{Theorem}[section]
\newtheorem{lemma}{Lemma}[section]

\newtheorem{corollary}{Corollary}[section]
\newtheorem{definition}{Definition}[section]
\newtheorem{example}{Example}[section]

\newcommand{\ie}{\emph{i.e. }}
\newcommand{\etal}{\textit{et al.}}
\newcommand{\m}[1]{\mathcal{#1}}

\def\defined{\: {\stackrel{\scriptscriptstyle \Delta}{=}} \: }

\date{}
\begin{document}

\title{Deterministic Network Model Revisited:\\An Algebraic Network Coding Approach}

\author{MinJi Kim, Elona Erez, Edmund M. Yeh, Muriel M\'edard
\thanks{Manuscript received on ..., and revised on .... The material in this paper was presented at the Information Theory and Applications Workshop (ITA), San Diego, CA, January 2010; and at the 2010 Allerton Conference on Communication, Control, and Computing, Urbana-Champaign, IL, September 2010.}
\thanks{M. Kim and M. M\'{e}dard are with the Department of Electrical Engineering and Computer Science, Massachusetts Institute of Technology, Cambridge, MA 02139 USA (e-mail: minjikim@mit.edu; medard@mit.edu).}
\thanks{E. Erez and E. M. Yeh are with the Department of Electrical Engineering, Yale University, New Haven, CT 06511 USA (e-mail: elona.erez@yale.edu; edmund.yeh@yale.edu).}
}
\maketitle

\begin{abstract}
The capacity of multiuser networks has been a long-standing problem in information theory. Recently, Avestimehr \etal\ have proposed a deterministic network model to approximate multiuser wireless networks. This model, known as the ADT network model, takes into account the broadcast nature of wireless medium and interference.

We show that the ADT network model can be described within the algebraic network coding framework introduced by Koetter and M\'{e}dard. We prove that the ADT network problem can be captured by a single matrix, and show that the min-cut of an ADT network is the rank of this matrix; thus, eliminating the need to optimize over exponential number of cuts between two nodes to compute the min-cut of an ADT network. We extend the capacity characterization for ADT networks to a more general set of connections, including single unicast/multicast connection and non-multicast connections such as multiple multicast, disjoint multicast, and two-level multicast. We also provide sufficiency conditions for achievability in ADT networks for any general connection set. In addition, we show that random linear network coding, a randomized distributed algorithm for network code construction, achieves the capacity for the connections listed above. Furthermore, we extend the ADT networks to those with random erasures and cycles (thus, allowing bi-directional links).

In addition, we propose an efficient linear code construction for the deterministic wireless multicast relay network model. Note that Avestimehr \etal's proposed code construction is not guaranteed to be efficient and may potentially involve an infinite block length. Unlike several previous coding schemes, we do not attempt to find flows in the network. Instead, for a layered network, we maintain an invariant where it is required that at each stage of the code construction, certain sets of codewords are linearly independent.
\end{abstract} 
\begin{keywords}
Network Coding, Deterministic Network, Algebraic Coding, Multicast, Non-multicast, Code Construction
\end{keywords}

\section{Introduction}\label{sec:Introduction}
Finding the capacity as well as the code construction for the multi-user wireless networks
are generally open problems. Even the relatively simple relay network with one source, one sink, and one relay, has not been fully characterized. There are two sources of disturbances in multi-user wireless networks -- channel noise and interference among users in the network. In order to better approximate the Gaussian multi-user wireless networks, \cite{adt1}\cite{adt2} proposed a binary linear deterministic network model (known as the ADT model), which takes into account the multi-user interference but not the noise. A node within the network receives the bit if the signal is above the noise level; multiple bits that simultaneously arrive at a node are superposed.

References \cite{adt1}\cite{adt2} showed that, for a multicast connection where a single source wishes to transmit the same data to a set of destinations, the achievable rate is equal to the minimal cut between the source and any of the destinations. Note that min-cut of an ADT network may not equal to the graph theoretical cut value, as we shall discuss in Section \ref{sec:mincut}. In addition, they showed that the minimal cut between the source and a destination is equal to the minimal rank of incidence matrices of all cuts between the two nodes. This can be viewed as the equivalent of the Min-cut Max-flow criterion in the network coding for wireline networks \cite{ahlswede}\cite{algebraic}. It has been shown that for several networks, the gap between the capacity of the deterministic ADT model and that of the corresponding Gaussian network is bounded by a constant number of bits, which does not depend on the specific channel fading parameters \cite{adt1}\cite{bresler08}\cite{avestimehr09}.

In this paper, we make a connection between the ADT network and network coding -- in particular, algebraic network coding introduced by Koetter and M\'{e}dard \cite{algebraic}. This paper is based on work from \cite{mm}\cite{ee0}\cite{ee}. Other approaches to operations in high SNR networks have been proposed \cite{analog_opt}, however, we do not compare these different approaches but build upon the given model proposed by \cite{adt1}\cite{adt2}. We show that
the ADT network problems, including that of computing the min-cut and constructing a code, can be captured by the algebraic network coding framework.

In the context of network coding, \cite{algebraic} showed that the solvability of the communication problem \cite{ahlswede} is equivalent to ensuring that a certain polynomial does \emph{not} evaluate to zero -- \ie avoid the roots of this certain polynomial. Furthermore, \cite{algebraic} showed that there are only \emph{a fixed finite number of roots} of the polynomial; thus, with large enough field size, decodability can be guaranteed even under randomized coding schemes as shown in \cite{rlc}. As we increase the field size $\mathbf{F}_q$, the space of feasible network codes increases exponentially; while the number of roots remain fixed.

We show that the solvability of ADT network problem can be characterized in a similar manner. The important difference between the algebraic network coding in \cite{algebraic} and the ADT network is that the broadcast as well as the interference constraints are embedded in the ADT network. Note that the interference constraint, represented by the additive multiple access channels (MAC), can be easily incorporated into the algebraic framework in \cite{algebraic} by pre-encoding at the transmitting nodes (\ie MAC users). This is due to the fact that the MAC is modeled using finite field additive channel; thus, the operations performed by the MAC can be ``canceled'' by the transmitter appropriately pre-encoding the packets.

On the other hand, the broadcast constraint may seem more difficult to incorporate, as the same code affects the outputs of the broadcast channel 
simultaneously, and the dependencies propagate through the network. Thus, in essence, this paper shows that this broadcast constraint is not problematic.

To briefly describe the intuition, consider an ADT network without the broadcast constraint -- \ie the broadcast edges do not need to carry the same information. Using this ``unconstrained'' version of the ADT network, the algebraic framework in \cite{algebraic} can be applied directly; thus, there is only a finite fixed number of roots that need to be avoided. Furthermore, as the field size increases, the probability of randomly selecting a root approaches zero. Now, we ``re-apply'' the broadcast constraints to this unconstrained ADT network. The broadcast constraint fixes the codes of the broadcast edges to be the same; this is equivalent to intersecting the space of network coding solutions with an hyperplane, which enforces the output ports of the broadcast to carry the same code. As shown in Figure \ref{fig:poly}, this operations does not change the polynomial whose root we have to avoid, but changes the hyperspace we operate in. As a result, this operation does not affect the roots of the polynomial; thus, there are still only a fixed finite number of roots that need to avoided, and with high enough field sizes, the probability of randomly selecting a root approaches zero. Note that intersecting the space of network coding solutions with an hyperplane may even ``remove'' some roots of the polynomial from consideration; therefore, we may effectively have fewer roots to avoid. By the same argument as \cite{algebraic}\cite{rlc}, we can then show that the solvability of an ADT network problem is equivalent to ensuring that a certain polynomial does not evaluate to zero within the space defined by the polynomial and the broadcast constraint hyperplane. As a result, we can describe the ADT network within the algebraic network coding framework and extend the random linear network coding results to the ADT networks.

\begin{figure}[tbp]
\begin{center}
\includegraphics[width=0.26\textwidth]{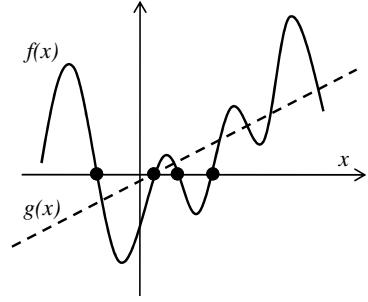}
\end{center}\vspace*{-.1cm}\caption{A polynomial $f(x)$ and a hyperplane $g(x)$ with one variable $x \in\mathbb{R}$. The black dots represent the roots of $f(x)$. When considering the space where $f(x)$ intersects $g(x)$, the hyperplane $g(x)$ limits the space in which $f(x)$ operates in; however, does not change the roots of $f(x)$. Some of the roots of $f(x)$ may no longer be ``feasible'' given the additional constraint; thus, this operation may reduce the number of roots that we have to consider.}\label{fig:poly}
\end{figure}

Using this insight, we prove that the ADT network problem can be captured by a single matrix, called the \emph{system matrix}. We show that the min-cut of an ADT network is
the rank of the system matrix; thus, eliminating the need to optimize over exponential number of cuts between two nodes to compute the min-cut of an ADT network. We extend the capacity characterization for ADT networks to a more general set of connections, including single unicast/multicast connection and non-multicast connections such as multiple multicast, disjoint multicast, and two-level multicast. We also provide sufficiency conditions for achievability in ADT networks for any general connection set. Furthermore, we extend the results on ADT networks to those with random erasures and cycles (thus, allowing bi-directional links).

We show that a direct consequence of this connection between ADT network problems and algebraic network coding is that random linear network coding, a randomized distributed algorithm for network code construction, achieves the capacity for the connections listed above. However, random linear network coding does not guarantee decodability; it allows decodability at all destinations with high probability.

Therefore, we propose an efficient linear code construction for multicasting in ADT networks that \emph{guarantees decodability}, if such code exists. Note that Avestimehr \etal's proposed code construction is not guaranteed to be efficient and may potentially involve an infinite block length. Unlike several previous coding schemes \cite{fragouli}\cite{sadegh09}\cite{goemans}, we do not attempt to find flows in the network. Instead, for a layered network, we maintain an invariant where it is required that at each stage of the code construction, certain sets of codewords are linearly independent. We assume that any node in the network can potentially be a destination. We design the code such that if the min-cut from the source to a certain node is at least the required rate, then the node will be able to reconstruct the data of the source using matrix inversion. In addition, when normalized by the number of sinks, our code construction has a complexity which is comparable to those of previous coding schemes for a single sink.

Our construction can be viewed as a non-straightforward generalization of the algorithm in \cite{jaggi05} for the construction of linear codes for multicast wireline networks. Each sink receives on its incoming edges a linear transformation of the source. The generalization of the code construction to the ADT network model is not straightforward, due to the broadcast constraint and the interference constraint, which are embedded into the ADT network model.

The paper is organized as follows. We present the network model in Section \ref{sec:model}, and an algebraic formulation of the ADT network in Section \ref{sec:algebraic}. Using this algebraic formulation, we provide a definition of the min-cut in ADT networks in Section \ref{sec:mincut}. In Sections \ref{sec:singlesource}, we restate the Min-cut Max-flow theorem using our algebraic formulation, and present new capacity characterizations for ADT networks to a more general set of traffic requirements in Section \ref{sec:general}. The results in Section \ref{sec:general} show the optimality of linear operations for non-multicast connections such as disjoint multicast and two-level multicast connections. In Section \ref{sec:robust}, we study ADT networks with link failures, and characterize the set of link failures such that the network solution is guaranteed to remain successful. Furthermore, in Sections \ref{sec:delay}, we extend the achievability results to ADT networks with delay. In Section \ref{sec:code}, we present our code construction algorithm for multicasting in ADT networks, and analyze its performance. Finally, we conclude in Section \ref{sec:conclusions}.

\section{Background}\label{sec:background}
Avestimeher \etal\ introduced the ADT network model to better approximate wireless networks \cite{adt1}\cite{adt2}. In the same work, they characterized the capacity of the ADT networks, and generalized the Min-cut Max-flow theorem for graphs to ADT networks for single unicast/multicast connections.

It has been shown that for several networks, the ADT network model approximates the capacity of the corresponding Gaussian network to within a constant number of bits. For instance, \cite{adt1} considered the single relay channel and the diamond network, and showed that the gap between the capacity of the ADT model and that of Gaussian network is within 1 bit and 2 bits, respectively. Reference \cite{bresler08} considered many-to-one and one-to-many Gaussian interference networks. The networks in \cite{bresler08} are special cases of interference network with multiple users, where the interference are either experienced (many-to-one) or caused by (one-to-many) a single user. It was shown that in these cases, the gap between the capacity of the Gaussian interference channel and the corresponding deterministic interference channel is again bounded by a constant number of bits. The work in \cite{bresler08} provided an alternative proof to \cite{etkin08} on the existence of a scheme that can achieve a constant gap from the capacity for all values of channel parameters. In \cite{avestimehr09}, the half-duplex butterfly network was considered. They showed that the deterministic model approximates the symmetric Gaussian butterfly network to within a constant.

As a result, there has been significant interest in finding an efficient code construction algorithm for the ADT network model.
In the case of unicast communication, a number of previous code constructions have been proposed for wireless relay networks. It is important to observe that in the code constructions for unicast communication, routing \cite{sadegh09} or one-bit operations \cite{fragouli} are sufficient for achieving the capacity of the deterministic model. Amaudruz and Fragouli \cite{amaudruz09} proposed an algorithm which can be viewed as an application of the Ford and Fulkerson flow construction to the deterministic model. The complexity of the algorithm was shown to be $O(|\m{V}||\m{E}|R^5)$, where $\m{V}$ is the set of nodes in the network, $\m{E}$ is the set of edges, and $R$ is the rate of the code. In \cite{sadegh09}, another algorithm for finding the flow for
unicast networks was developed. The algorithm is based on an extension of the Rado-Hall transversal theorem for matroids and on
Edmonds' theorem. The transmission scheme in \cite{sadegh09} extracts at each relay node a subset of the input vectors and sets
the outputs to the same values as that subset. In \cite{goemans}, it was shown that the deterministic model can be
viewed as a special case of a more abstract flow model, called \emph{linking network}, which is based on
linking systems and matroids. Using this approach, \cite{goemans} achieved a code complexity $O(\lambda N_{layer}^3\log N_{layer})$, where $\lambda$ is the number of layers in the layered network, and $N_{layer}$ is the maximal number of nodes in a layer.  Note that linear network coding is known to be matroidal \cite{matroids}; thus, the fact that ADT networks are matroidal \cite{goemans} is consistent with our result.

In the case for multicast communication, however, routing or one-bit operations may not be sufficient to achieve the capacity in the ADT model. This can be shown by considering the example in Figure \ref{example_network}, which is given in \cite{feder03}\cite{lehman04}\cite{fragouli04} for network coding. From the analysis for network coding, it follows that in the case of the deterministic model, the maximal rate $2$ can be achieved simultaneously for all sinks only with an alphabet size which is at least $3$. To see this, observe that to achieve rate $2$ the source has to transmit at its outputs two statistically independent symbols $x_1,x_2$. For node $v^2_i,1\leq i\leq 4$ at the second layer, the transmitted symbol is a certain function  of the symbols $x_1,x_2$, given by $y_i=f_i(x_1,x_2)$. Node $v^3_i,1\leq i\leq 4$ at the third layer transmits at its outputs two functions of $y_i$, given by $f^1_i(y_i)$, $f^2_i(y_i)$. Sink $t_i,1\leq i\leq 6$ receives at its two inputs symbols of the form $f^1_j(y_j),f^2_j(y_j)+f^1_k(y_k)$ for some $1\leq j\leq 4,1\leq k\leq 4,j\neq k$. It follows that without rate loss, we can always assume $f^1_j(y_j)=y_j$ for each $1\leq j\leq 4$. In that case, the sink $t_i$ receives $y_j$ at its upper input and can therefore find $ f^2_j(y_j)$ and eliminate it from its second received symbol. Thus, it is equivalent to the situation in which the sink receives $y_j,y_k$. This in turn is exactly the situation in \cite{lehman04} (Theorem 3.1) for network coding. Since the channels are all binary in the deterministic model, it follows that the minimal required alphabet size is in fact $2^2 = 4$, and therefore the minimal vector length is $\log_2(4)=2$. Thus, for multicasting in ADT networks, we need to either operate in a higher field size, $\mathbb{F}_q$, $q \geq 2$, or use vector coding (or both).

\begin{figure}[tbp]
\begin{center}
\includegraphics[width=0.4\textwidth]{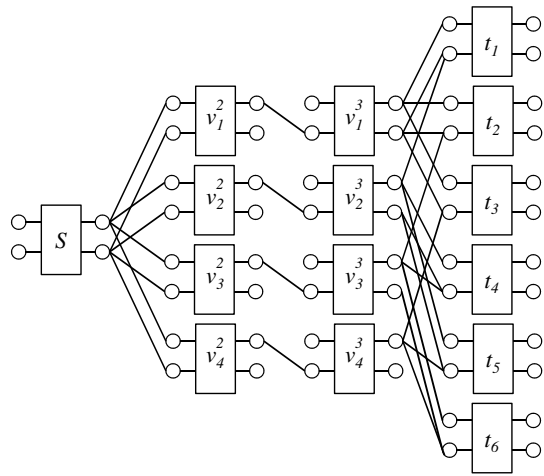}
\end{center}\vspace*{-.3cm}\caption{Example network for a non-binary code.}\label{example_network}
\end{figure}

References \cite{fragouli_itw}\cite{fragouli_vector}, independently from \cite{mm}\cite{ee0}\cite{ee}, proposed a polynomial time algorithm for multicasting in ADT networks. In particular, \cite{fragouli_vector} extended the algebraic network coding result \cite{algebraic} to vector network coding, and showed that constructing a valid vector code is equivalent to certain algebraic conditions. This result \cite{fragouli_vector} is supported by the result from \cite{jaggi}. Reference \cite{jaggi} introduced network codes, called \emph{permute-and-add}, that only require bit-wise vector operations to take advantage of low-complexity operations in $\mathbb{F}_2$. In addition, \cite{jaggi} showed that codes in higher field size $\mathbb{F}_q$ can be mapped to binary-vector codes without loss in performance. This insight, combined with that of \cite{algebraic}, suggests that an algebraic property of a scalar code may translate into another algebraic property of the corresponding vector code.


\section{Network Model}\label{sec:model}
\begin{figure}[tbp]
\begin{center}
\includegraphics[width=0.40\textwidth]{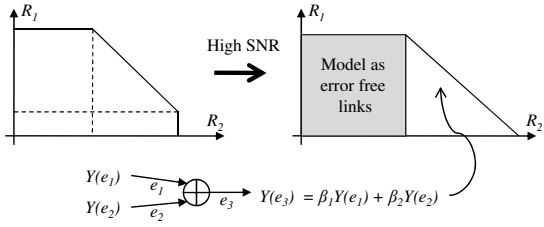}
\end{center}\vspace*{-.1cm}\caption{Additive MAC with two users, and the corresponding rate region. The triangular region is modeled as a set of finite field additive MACs.}\label{fig:mac}
\end{figure}

As in \cite{adt1}\cite{adt2}, we shall consider the high SNR regime, in which interference is the dominating factor. In high SNR, analog network coding, which allows/encourages strategic interference, is near optimal \cite{analog_opt}. Analog network coding is a physical layer coding technique, introduced by \cite{analogNC}, in which intermediate nodes amplify-and-forward the received signals without decoding. Thus, the nodes amplify not only the superposed signals from different transmitters but also the noise. Note that a network operating in high SNR regime is different from a network with high gain since a large gain amplifies the noise as well as the signal.


In the high SNR regime, the Cover-Wyner region may be well approximated by the combination of two regions, one square and one triangular, as in Figure \ref{fig:mac}. The square (shaded) part can be modeled as parallel links for the users, since they do not interfere. The triangular (unshaded) part can be considered as that of a time-division multiplexing (TDM), which is equivalent to using noiseless finite-field additive MAC \cite{mac}. This result holds not only for binary field additive MAC, but also for higher field size additive MAC \cite{mac}.

The ADT network model uses binary channels, and thus, binary additive MACs are used to model interference. Prior to \cite{adt1}\cite{adt2}, Effros \etal \ presented an additive MAC over a finite field $\mathbb{F}_q$ \cite{additive}. The Min-cut Max-flow theorem holds for all of the cases above. It may seem that the ADT network model differs greatly from that of \cite{additive} owing to the difference in field sizes used. In general, codes in $\mathbf{F}_q$ subsume binary codes, \ie \emph{binary-vector} codes in $(\mathbf{F}_2)^m$. However, for point-to-point links with memory (or equivalently by allowing nodes to code across time), we can convert a code in $(\mathbf{F}_2)^m$ to a code in higher field size $\mathbf{F}_q$ and vice versa by normalizing to an appropriate time unit.
 Note that ADT network model uses binary additive MACs and point-to-point links. Therefore, our work in part shows an equivalence of higher field size codes and binary-vector codes in ADT networks.

As noted in Section \ref{sec:background}, \cite{jaggi} presented a method of converting between binary-vector codes and higher field size codes. We can achieve a higher field size in ADT networks by combining multiple binary channels.
In other words, consider two nodes $V_1$ and $V_2$ with two binary channels connecting $V_1$ to $V_2$. Now, instead of considering them as two binary channels, we can ``combine'' the two channels as one with capacity of 2-bits. In this case, instead of using $\mathbb{F}_2$, we can use a larger field size of $\mathbb{F}_4$. Thus, selecting a larger field size $\mathbb{F}_q$, $q >2$ in ADT network results in fewer but higher capacity parallel channels. Reference \cite{jaggi} also provides a conversion from a code in a higher field $\mathbb{F}_q$ to a binary-vector scheme in $\mathbb{F}_{2^m}$ where $q\leq 2^m$. Therefore, a solution in $\mathbb{F}_q$ may be converted back to a binary-vector scheme, which may be more appropriate for the original ADT model. Furthermore, it is known that to achieve capacity for multicast connections, $\mathbb{F}_2$ is not sufficient \cite{edmund}; thus, we need to operate in a higher field size. Therefore, we shall not restrict ourselves to $\mathbb{F}_2$.

\begin{figure}[tbp]
\begin{center}
\includegraphics[width=0.45\textwidth]{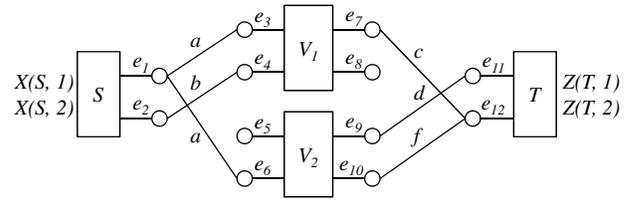}
\end{center}\vspace*{-.3cm}\caption{Example network. We omit $I(S)$ and $O(T)$ in this diagram as they do not participate in the communication.}\label{fig:network}
\end{figure}

\begin{figure}[tbp]
\begin{center}
\includegraphics[width=0.31\textwidth]{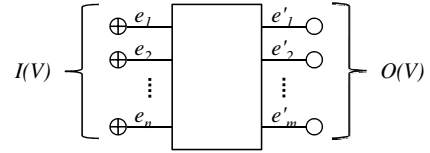}
\end{center}\vspace*{-.1cm}\caption{A supernode $V$.}\label{fig:supernode}
\end{figure}

\begin{figure}[tbp]
\begin{align*}
Y(e_1) & = \alpha_{(1, e_1)}X(S, 1) + \alpha_{(2, e_1)}X(S, 2)\\
Y(e_2) & = \alpha_{(1, e_2)}X(S, 1) + \alpha_{(2, e_2)}X(S, 2)\\
Y(e_3) &= Y(e_6) = Y(e_1)\\
Y(e_4) &= Y(e_2)\\
Y(e_5) & = Y(e_{8}) = 0\\
Y(e_{7}) & = \beta_{(e_3, e_{7})}Y(e_3) + \beta_{(e_4, e_{7})}Y(e_4)\\
Y(e_{9}) & = Y(e_{11}) =  \beta_{(e_6, e_{9})}Y(e_6)\\
Y(e_{10}) & = \beta_{(e_6, e_{10})}Y(e_6)\\
Y(e_{12}) & = Y(e_{7}) + Y(e_{10})\\
Z(T, 1) & = \epsilon_{(e_{11}, (T, 1))} Y(e_{11}) +  \epsilon_{(e_{12}, (T, 1))} Y(e_{12})\\
Z(T, 2) & = \epsilon_{(e_{11}, (T, 2))} Y(e_{11}) +  \epsilon_{(e_{12}, (T, 2))} Y(e_{12})
\end{align*}\vspace*{-.2cm}\caption{Equations relating the various processes of Figure \ref{fig:network}.}\label{fig:equations}
\end{figure}

\begin{figure*}[tbp]
\begin{center}
\includegraphics[width=0.70\textwidth]{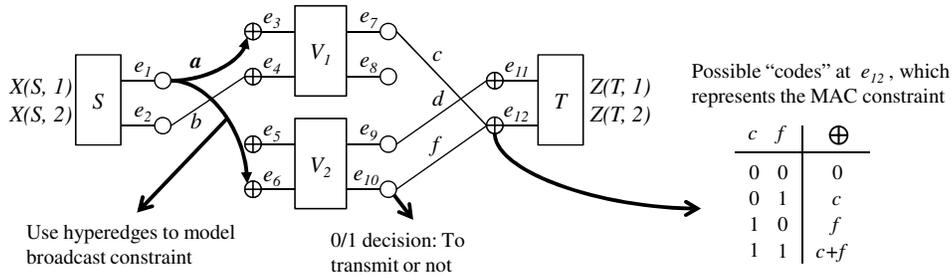}
\end{center}\vspace*{-.2cm}\caption{A new interpretation of the example network from Figure \ref{fig:network}. The broadcast channel is modeled using an \emph{hyperedge}. As a result, an output port's decision to transmit or not naturally affects all the input ports adjacent to it. Furthermore, interference is modeled using a finite field additive MAC, which provides a set of possible \emph{binary codes} at the input ports.
}\label{fig:newnetwork}
\end{figure*}
We now proceed to defining the network model precisely. A wireless network is modeled using a directed graph $G = (\m{V}, \m{E})$ with a \emph{supernode} set $\m{V}$ and an edge set $\m{E}$, as shown in Figure \ref{fig:network}. A supernode $V\in \m{V}$ is a node of the original network. We use the term supernode to emphasize the fact that supernode $V$ consists of \emph{input ports} $I(V)$ and \emph{output ports} $O(V)$, as shown in Figure \ref{fig:supernode}. Let $\m{S}, \m{T} \subseteq \m{V}$ be the set of source and destination supernodes. An edge $(e_1, e_2)$ may exist from an output port $e_1 \in O(V_1)$ to an input port $e_2 \in I(V_2)$, for any $V_1, V_2 \in \m{V}$. Let $\m{E}(V_1, V_2)$ be the set of edges from $O(V_1)$ to $I(V_2)$. All edges are of unit capacity, where capacity is normalized with respect to the symbol size of $\mathbb{F}_q$.

Noise is embedded, or \emph{hard-coded}, in the structure of the ADT network in the following way. Parallel links of $\m{E}(V_1, V_2)$ deterministically model noise between $V_1$ and $V_2$.
Let $SNR_{(V_i, V_j)}$ be the signal-to-noise ratio from supernode $V_i$ to supernode $V_j$. Then, $|\m{E}(V_1, V_2)| = \lceil \frac{1}{2} \log SNR_{(V_i, V_j)}\rceil$. Thus, the number of edges between two supernodes $V_i$ and $V_j$ represents the channel quality (equivalently, the noise) between the two supernodes.

Given such a wireless network $G = (\m{V}, \m{E})$, let $\m{S}$ be the set of sources. A source supernode $S \in \m{S}$ has \emph{independent} random processes $\m{X}(S) = [X(S, 1), X(S, 2), ..., X(S, \mu(S))]$, $\mu(S) \leq |O(S)|$, which it wishes to communicate to a set of destination supernodes $\m{T}(S) \subseteq \m{T}$. In other words, we want $T \in \m{T}(S)$ to replicate a subset of the random processes, denoted $\m{X}(S, T) \subseteq \m{X}(S)$, by the means of the network. Note that the algebraic formulation is not restricted to multicast connections; different sources may wish to communicate to different subsets of destinations. We define a \emph{connection} $c$ as a triple $(S, T, \m{X}(S, T))$, and the rate of $c$ is defined as $R(c) = \sum_{X(S, i) \in \m{X}(S, T)} H(X(S, i)) = |\m{X}(S,T)|$ (symbols).

Information is transmitted through the network in the following manner. A supernode $V$ sends information through $e \in O(V)$ at a rate at most one symbol per time unit. Let $Y(e)$ denote the random process at port $e$. In general, $Y(e)$, $e\in O(V)$, is a function of $Y(e')$, $e'\in I(V)$. In this paper, we consider only linear functions.
\begin{equation}\label{eq:y}
Y(e) = \sum_{e'\in I(V)} \beta_{(e', e)} Y(e'), \text{ for $e \in O(V)$.}
\end{equation}
For a source supernode $S$, and $e\in O(S)$,
\begin{equation}\label{eq:y-s}
Y(e) =\sum_{e'\in I(V)} \beta_{(e', e)} Y(e') + \sum_{X(S,i) \in \m{X}(S)} \alpha_{(i, e)} X(S, i).
\end{equation}
Finally, the destination $T$ receives a collection of input processes $Y(e')$, $e' \in I(T)$. Supernode $T$ generates a set of random processes $\m{Z}(T)= [Z(T, 1), Z(T, 2), ..., Z(T, \nu(T))]$ where
\begin{equation}\label{eq:z}
Z(T, i) = \sum_{e' \in I(T)} \epsilon_{(e', (T, i))} Y(e').
\end{equation}
A connection $c = (S, T, \m{X}(S, T))$ is established successfully if $\m{X}(S) = \m{Z}(T)$. A supernode $V$ is said to \emph{broadcast} to a set $\m{V'}\subseteq \m{V}$ if $\m{E}(V, V') \ne \emptyset$ for all $V'\in  \m{V'}$. In Figure \ref{fig:network}, $S$ broadcasts to supernodes $V_1$ and $V_2$. Superposition occurs at the input port $e' \in I(V)$, \ie $Y(e') = \sum_{(e, e') \in \m{E}} Y(e)$ over a finite field $\mathbb{F}_q$. We say there is a $|\m{V'}|$-user MAC channel if $\m{E}(V', V) \ne \emptyset$ for all $V' \in \m{V'}$. In Figure \ref{fig:network}, supernodes $V_1$ and $V_2$ are users, and $T$ the receiver in a 2-user MAC.

For a given network $G$ and a set of connections $\m{C}$, we say that $(G,\m{C})$ is \emph{solvable} if it is possible to establish successfully all connections $c\in \m{C}$. The broadcast and MAC constraints are given by the network; however, we are free to choose the variables $\alpha_{(i, e)}$, $\beta_{(e', e)}$, and $\epsilon_{(e',i)}$ from $\mathbb{F}_q$. Thus, the problem of checking whether a given $(G, \m{C})$ is solvable is equivalent to finding a feasible assignment to $\alpha_{(i,e)}, \beta_{(e', e)}$, and $\epsilon_{(e', (T, i))}$.

\begin{example}\label{ex:equations} The equations in Figure \ref{fig:equations} relate the various processes in the example network in Figure \ref{fig:network}. Note that in Figure \ref{fig:network}, we have set $Y(e_1) =a$, $Y(e_2)=b$, $Y(e_7)=c$, $Y(e_9) = d$, and $Y(e_{10})= f$ for notational simplicity.
\end{example}

\subsection{An Interpretation of the Network Model}\label{sec:interpretation}

\begin{figure}[tbp]
\begin{center}
\includegraphics[width=0.23\textwidth]{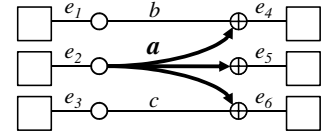}
\end{center}\vspace*{-.1cm}\caption{An example of finite field additive MAC.}\label{fig:constraint}
\end{figure}

The ADT network model uses multiple channels from an output port to model broadcast. In Figure \ref{fig:network}, there are two edges from output port $e_1$ to input ports $e_3$ and $e_6$; however, due to the broadcast constraint, the two edges $(e_1, e_3)$ and $(e_1, e_6)$ carry the same information $a$. This introduces considerable complexity in constructing a network code as well as computing the min-cut of the network \cite{adt1}\cite{adt2}\cite{fragouli}\cite{goemans}. This is because multiple edges from a port do not capture the broadcast dependencies. Furthermore, the broadcast dependencies have to be propagated through the network.

In our approach, we remedy this by introducing the use of hyperedges, as shown in Figure \ref{fig:newnetwork} and Section \ref{sec:algebraic}. An output port's decision to transmit affects the entire hyperedge; thus, the output port transmits to all the input ports connected to the hyperedge simultaneously. This removes the difficulties of computing the min-cut of ADT networks (Section \ref{sec:mincut}), as it naturally captures the broadcast dependencies.

The finite field additive MAC model can be viewed as a set of codes that an input port may receive. As shown in Figure \ref{fig:newnetwork}, input port $e_{12}$ receives one of the four possible codes. The code that $e_{12}$ receives depends on output ports $e_7$'s and $e_{10}$'s decision to transmit or not.

The difficulty in constructing a network code does not come from any single broadcast or MAC constraint. The difficulty in constructing a code is in satisfying multiple MAC and broadcast constraints simultaneously. For example, in Figure \ref{fig:constraint}, the fact that $e_4$ may receive $a+b$ does not constrain the choice of $a$ nor $b$. This is because we can choose any $a$ and $b$ such that $a+b \ne 0$, and ensure that both $a$ and $b$ are decoded as long as enough degrees of freedom are received by the destination node. The same argument applies to $e_6$ receiving $a+c$. However, the problem arises from the fact that a choice of value for $a$ at $e_4$ interacts both with $b$ and $c$. In such a case, we need to ensure that both $a + b \ne 0$ and $a + c \ne 0$; thus, our constraint is $(a+b)(a+c)\ne 0$. As the network grows in size, we will need to satisfy more constraints simultaneously. As we shall see in Section \ref{sec:mincut}, we eliminate this difficulty by allowing the use of a larger field, $\mathbb{F}_q$.

\section{Algebraic Network Coding Formulation}\label{sec:algebraic}

We provide an algebraic formulation for the ADT network problem $(G, \m{C})$, and present an algebraic condition under which the system $(G, \m{C})$ is solvable. We assume that $G$ is acyclic in this section; however, we shall extend the results in this section to ADT networks with cycles in Section \ref{sec:delay}. For simplicity, we describe the multicast problem with a single source $S$ and a set of destination supernodes $\m{T}$, as in Figure \ref{fig:multicast}. However, this formulation can be extended to multiple source $S_1, S_2, ... S_K$ by adding a super-source $S$ as in Figure \ref{fig:supersource}.

We define a system matrix $M$ to describe the relationship between the source's random processes $\m{X}(S)$ and the destinations' processes $\m{Z} = [\m{Z}(T_1), \m{Z}(T_2), ..., \m{Z}(T_{|\m{T}|})]$. Thus, we want to characterize $M$ where
\begin{equation}
\m{Z} = \m{X}(S) \cdot M.
\end{equation}
The matrix $M$ is composed of three matrices, $A$, $F$, and $B$.

\subsection{Adjacency matrix $F$}\label{sec:transfermatrices}

Given $G$, we define the adjacency matrix $F$ as follows:
\begin{equation}\label{eq:F}F_{i, j}= \begin{cases}
1 & \text{if $(e_i, e_j) \in \m{E}$,}\\
\beta_{(e_i, e_j)} & \text{if $e_i \in I(V)$, $e_j \in O(V)$ for $V \in \m{V}$,}\\
0 & \text{otherwise.}
\end{cases}
\end{equation}

\begin{figure}[tbp]
\begin{center}
\includegraphics[width=0.38\textwidth]{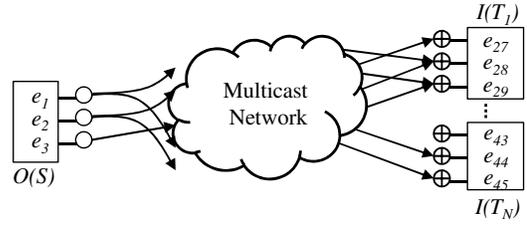}
\end{center}\vspace*{-.3cm}\caption{Single multicast network with source $S$ and receivers $T_1, ..., T_N$.}\label{fig:multicast}
\end{figure}

Matrix $F$ is defined on the ports, rather than on the supernodes. This is because, in the ADT model, each port is the basic receiver/transmitter unit. Each entry $F_{i,j}$ represents the input-output relationships of the ports. A zero entry indicates that the ports are not directly connected, while an entry of one represents that they are connected. The adjacency matrix $F$ naturally captures the physical structure of the ADT network. Note that a row with multiple entries of 1 represents the broadcast hyperedge; while a column with multiple entries of 1 represents the MAC constraint. Note that the 0-1 entries of $F$ represent the \emph{fixed} network topology as well as the broadcast and MAC constraints. On the other hand, $\beta_{(e_i, e_j)}$ are free variables, representing the coding coefficients used at $V$ to map the input port processes to the output port processes. This is the key difference between the work presented here and in \cite{algebraic} -- $F$ is partially fixed in the ADT network model due to network topology and broadcast/MAC constraints, while in \cite{algebraic}, only the network topology affects $F$.

In \cite{adt1}\cite{adt2}, the supernodes are allowed to perform any internal operations; while in \cite{fragouli}\cite{goemans}, only permutation matrices (\ie routing) are allowed. Note that \cite{fragouli}\cite{goemans} only consider a single unicast traffic, in which routing is known to be sufficient. References \cite{adt1}\cite{adt2} showed that linear operations are sufficient to achieve the capacity in ADT networks for a single multicast traffic. We consider a general setup in which $\beta_{(e_i, e_j)} \in \mathbb{F}_q$ -- thus, allowing any matrix operation, as in \cite{adt1}\cite{adt2}.

\begin{figure}[tbp]
\begin{center}
\includegraphics[width=.48\textwidth]{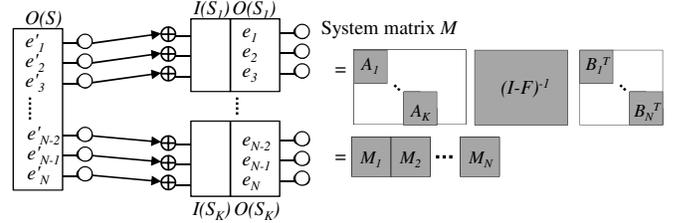}
\end{center}\vspace*{-.1cm}\caption{A network with multiple sources $S_1, S_2, ..., S_K$ can be converted to a single source problem by adding a super-source $S$ with $|O(S)| = \sum_{i=1}^K |O(S_i)|$. Each $e'_j \in O(S)$ has a ``one-to-one connection'' to a $e_j \in O(S_i)$, for $i \in [1, K]$. Matrix $A_i$ represents the encoding matrix for source $S_i$, while $B_j$ is the decoding matrix at destination $T_j$. The white area represents the zero elements, and the shaded area represents the coding coefficients.
}\label{fig:supersource}
\end{figure}

Note that $F^k$, the $k$-th power of an adjacency matrix of a graph $G$, shows the existence of paths of length $k$ between any two nodes in $G$. Therefore, the series $I + F + F^2 + F^3 + ...$ represents the connectivity of the network. It can be verified that $F$ is nilpotent, which means that there exists a $k$ such that $F^k$ is a zero matrix. As a result, $I + F + F^2 + F^3 + ...$ can be written as $(I-F)^{-1}$. Thus, $(I-F)^{-1}$ represents the impulse response of the network. Note that, $(I-F)^{-1}$ exists for all acyclic network since $I-F$ is an upper-triangle matrix with all diagonal entries equal to 1; thus, $\det(I-F) = 1$.


\begin{example}\label{ex:F}
In Figure \ref{fig:F}, we provide the $12 \times 12$ adjacency matrix $F$ for the example network in Figures \ref{fig:network} and \ref{fig:newnetwork}. Note that the first row (with two entries of 1) represents the broadcast hyperedge, $e_1$ connected to both $e_3$ and $e_6$. The last column with two entries equal to 1 represents the MAC constraint, both $e_7$ and $e_{10}$ transmitting to $e_{12}$.
The highlighted elements in $F$ represent the coding variables, $\beta_{(e', e)}$, of $V_1$ and $V_2$. 
For some $(e', e)$, $\beta_{(e',e)} =0$ since these ports of $V_1$ and $V_2$ are not used.
\end{example}

\begin{figure}[tbp]
\[\scriptsize
\left(
\begin{array}{cccccccccccc}
0 & 0 & 1 & 0 & 0 & 1 & 0 & 0 & 0 & 0 & 0 & 0\\
0 & 0 & 0 & 1 & 0 & 0 & 0 & 0 & 0 & 0 & 0 & 0\\
0 & 0 & 0 & 0 & 0 & 0 & \cellcolor[gray]{.8} \beta_{(e_3, e_7)} & \cellcolor[gray]{.8} 0 & 0 & 0 & 0 & 0\\
0 & 0 & 0 & 0 & 0 & 0 & \cellcolor[gray]{.8} \beta_{(e_4, e_7)} & \cellcolor[gray]{.8} 0 & 0 & 0 & 0 & 0\\
0 & 0 & 0 & 0 & 0 & 0 & 0 & 0 & \cellcolor[gray]{.8} 0 & \cellcolor[gray]{.8} 0 & 0 & 0\\
0 & 0 & 0 & 0 & 0 & 0 & 0 & 0 & \cellcolor[gray]{.8} \beta_{(e_6, e_{9})} & \cellcolor[gray]{.8} \beta_{(e_6, e_{10})}& 0 & 0\\
0 & 0 & 0 & 0 & 0 & 0 & 0 & 0 & 0 & 0 & 0 & 1\\
0 & 0 & 0 & 0 & 0 & 0 & 0 & 0 & 0 & 0 & 0 & 0\\
0 & 0 & 0 & 0 & 0 & 0 & 0 & 0 & 0 & 0 & 1 & 0\\
0 & 0 & 0 & 0 & 0 & 0 & 0 & 0 & 0 & 0 & 0 & 1\\
0 & 0 & 0 & 0 & 0 & 0 & 0 & 0 & 0 & 0 & 0 & 0\\
0 & 0 & 0 & 0 & 0 & 0 & 0 & 0 & 0 & 0 & 0 & 0\\
\end{array}
\right)
\]\vspace*{-.1cm}\caption{$12 \times 12$ adjacency matrix $F$ for network in Figure \ref{fig:network}.}\label{fig:F}
\end{figure}

\subsection{Encoding matrix $A$}\label{sec:encodingmatrix}

Matrix $A$ represents the encoding operations performed at $S$. We define a $|\m{X}(S)| \times |\m{E}|$ encoding matrix $A$ as follows:
\begin{equation}
A_{i,j} = \begin{cases}
\alpha_{(i, e_j)}  & \text{if $e_j\in O(S)$ and $X(S,i) \in \m{X}(S)$,}\\
0& \text{otherwise}.
\end{cases}
\end{equation}

\begin{example}\label{ex:A}
We provide the $2 \times 12$ encoding matrix $A$ for the network in Figure \ref{fig:network}.
\[
A = \begin{pmatrix}
\alpha_{1, e_1} & \alpha_{1, e_2} & 0 & \dotsb & 0\\
\alpha_{2, e_1} & \alpha_{2, e_2} & 0 & \dotsb & 0\\
\end{pmatrix}.
\]
\end{example}

\subsection{Decoding matrix $B$}\label{sec:decodingmatrix}

Matrix $B$ represents the decoding operations performed at the destination $T \in \m{T}$. Since there are $|\m{T}|$ destination nodes, $B$ is a matrix of size $|\m{Z}| \times |\m{E}|$ where $\m{Z}$ is the set of random processes derived at the destination supernodes. We define the decoding matrix $B$ as follows:
\begin{equation}
B_{i,(T_j, k)} = \begin{cases}
\epsilon_{(e_i, (T_j, k))} & \text{if $e_i \in I(T_j), Z(T_j, k) \in \m{Z}(T_j)$},\\
0 & \text{otherwise.}
\end{cases}
\end{equation}

\begin{example}\label{ex:B}
We provide the $2 \times 12$ decoding matrix $B$ for the example network in Figure \ref{fig:network}.
\[
B = \begin{pmatrix}
0 & \dotsb & 0 & \epsilon_{(e_{11}, (T, 1))}  & \epsilon_{(e_{12}, (T, 1))} \\
0 & \dotsb & 0 & \epsilon_{(e_{11}, (T, 2))}  & \epsilon_{(e_{12}, (T, 2))} \\
\end{pmatrix}.
\]
\end{example}

\subsection{System matrix $M$}\label{sec:systemmatrix}

\begin{theorem}\label{thm:m}
Given a network $G = (\m{V}, \m{E})$, let $A$, $B$, and $F$ be the encoding, decoding, and adjacency matrices, respectively. Then, the system matrix $M$ is given by
\begin{equation}
M = A (1-F)^{-1} B^T.
\end{equation}
\end{theorem}
\begin{proof}
The proof of this theorem is similar to that of Theorem 3 in \cite{algebraic}. As previously mentioned, $(I-F)^{-1} = (I + F + F^2 + ...)$ always exists for an acyclic network $G$.
\end{proof}

Note that the algebraic framework shows a clear separation between the given physical constraints (fixed 0-1 entries of $F$ showing the topology and the broadcast/MAC constraints), and the coding decisions. As mentioned previously, we can freely choose the coding variables $\alpha_{(i, e_j)}$, $\epsilon_{(e_i, (T_j, k))}$, and $\beta_{(e_i, e_j)}$. Thus, solvability of $(G, \m{C})$ is equivalent to assigning values to $\alpha_{(i, e_j)}$, $\epsilon_{(e_i, (T_j, k))}$, and $\beta_{(e_i, e_j)}$ such that each receiver $T \in \m{T}$ is able to decode the data it is intended to receive.

\begin{example}
We can combine the matrices $F$, $A$, and $B$ from Examples \ref{ex:F}, \ref{ex:A}, and \ref{ex:B} respectively to obtain the system matrix $M = A(I-F)^{-1}B^T$ for the network in Figure \ref{fig:network}. We show a schematic of the system matrix $M$ in Figure \ref{fig:systemmatrix}.
\end{example}

\begin{figure}[tbp]
\begin{center}
\includegraphics[width=0.45\textwidth]{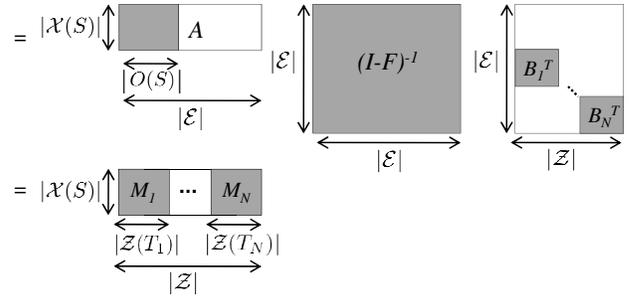}\vspace*{-.15cm}
\caption{The system matrix $M$ and its components $A$, $(I-F)^{-1}$, and $B$ for a single multicast connection with source $S$ and destinations $T_i$, $i\in [1, N]$.
}\label{fig:systemmatrix}
\end{center}
\end{figure}

\section{Definition of Min-cut}\label{sec:mincut}

Consider a source $S$ and a destination $T$. Reference \cite{adt1} proves the maximal achievable rate to be the minimum value of all $S$-$T$ cuts, denoted $mincut(S,T)$, which we reproduce below in Definition \ref{def:mincut}.

\begin{definition}[\textbf{Min-cut }\emph{\cite{adt1}\cite{adt2}}]\label{def:mincut} A cut $\Omega$ between a source $S$ and a destination $T$ is a partition of the supernodes into two disjoint sets $\Omega$ and $\Omega^c$ such that $S \in \Omega$ and $T \in \Omega^c$. For any cut, $G_{\Omega}$ is the incidence matrix associated with the bipartite graph with ports in $\Omega$ and $\Omega^c$. Then, the capacity of the given ADT network (equivalently, $mincut(S,T)$) is defined as
\[
mincut(S,T) = \min_{\Omega} rank(G_{\Omega}).
\]
This rate of $mincut(S,T)$ can be achieved using linear operations for a single unicast/multicast connection.\hspace*{1cm} $\blacksquare$
\end{definition}

Note that, with the above definition, in order to compute $mincut(S,T)$, we need to optimize over all cuts between $S$ and $T$. In addition, the proof of achievability in \cite{adt1} is not constructive, as it assumes infinite block length and does not consider the details of internal supernode operations.

We introduce a new algebraic definition of the min-cut, and show that it is equivalent to that of Definition \ref{def:mincut}.

\begin{theorem}\label{thm:mincut}
The capacity of the given ADT, equivalently the minimum value of all $S-T$ cuts $mincut(S,T)$, is
\begin{align*}
mincut(S,T) &= \min_{\Omega} \text{rank}(G_{\Omega})\\
&= \max_{\alpha_{(i, e)}, \beta_{(e', e)},\epsilon_{(e',i)}} \text{rank}(M).
\end{align*}
\end{theorem}
\begin{proof}
By \cite{adt1}, we know that $mincut(S,T) = \min_{\Omega} \text{rank}(G_{\Omega})$. Therefore, we show that $\max_{\alpha, \beta,\epsilon} \text{rank}(M)$ is equivalent to the maximal achievable rate in an ADT network.

First, we show that $mincut(S,T) \geq  \max_{\alpha, \beta,\epsilon} \text{rank}(M)$. In our algebraic formulation, $\m{Z}(T) = \m{X}(S)M$; thus, the rank of $M$ represents the rate achieved. Let $R = \max_{\alpha, \beta,\epsilon} \text{rank}(M)$. Then, there exists an assignment of $\alpha_{(i, e)}, \beta_{(e', e)},$ and $\epsilon_{(e',i)}$ such that the network achieves a rate of $R$. By the definition of min-cut, it must be the case that $mincut(S,T) \geq R$.

Second, we show that $mincut(S,T) \leq  \max_{\alpha, \beta,\epsilon} \text{rank}(M)$. Assume that $R = mincut(S,T)$. Then, by \cite{adt1}\cite{adt2}, there exists a linear configuration of the network such that we can achieve a rate of $R$ such that the destination $T$ is able to reproduce $\m{X}(S,T)$. This configuration of the network provides a linear relationship of the source-destination processes (actually, the resulting system matrix is an identity matrix); thus, an assignment of the variables $\alpha_{(i, e)}, \beta_{(e', e)}$, and $\epsilon_{(e',i)}$ for our algebraic framework. We denote $M'$ to be the system matrix corresponding to this assignment. Note that, by the definition, $M'$ is an $R\times R$ matrix with a rank of $R$. Therefore, $\max_{\alpha, \beta,\epsilon} \text{rank}(M) \geq \text{rank}(M') = mincut(S,T)$.
\end{proof}

The system matrix $M$ (thus, the network and decodability at the destinations)
depends not only on the structure of the ADT network, but also on the field size used, supernodes' internal operations, transmission rate, and connectivity. For example, the network topology may change with a choice of larger field size, since larger field sizes result in fewer parallel edges/channels. Another example, if we adjust the rate such that $|\m{X}(S)| \leq mincut(S,T)$, then $M$ is full rank. However, if $|\m{X}(S)|> mincut(S,T)$, then $M$ may have rank of $mincut(S,T)$ but not be full-rank. It is important to note that the cut value in the ADT network may not equal to the graph theoretical cut value (see Figure 2 in \cite{fragouli}).

\section{Min-cut Max-flow Theorem}\label{sec:singlesource}

In this section, we provide an algebraic interpretation of the Min-cut Max-flow theorem for a single unicast connection and a single multicast connection \cite{adt1}\cite{adt2}. This result is a direct consequence of \cite{algebraic} when applied to the algebraic formulation for the ADT network. We also show that a distributed randomized coding scheme achieves the capacity for these connections.

\begin{theorem}[Min-cut Max-flow Theorem] Given an acyclic network $G$ with a single connection $c = (S, T, \m{X}(S, T))$ of rate $R(c) = |\m{X}(S, T)|$, the following are equivalent:
\begin{enumerate}
\item A unicast connection $c$ is feasible.
\item $mincut(S,T) \geq R(c)$.
\item There exists an assignment of $\alpha_{(i, e_j)}$, $\epsilon_{(e_i, (T_j, k))}$, and $\beta_{(e_i, e_j)}$ such that the $R(c) \times R(c)$ system matrix $M$ is invertible in $\mathbb{F}_q$ (\ie $\det(M) \ne 0$).
\end{enumerate}\label{thm:mincut_maxflow}
\end{theorem}

\begin{proof}
Statements 1) and 2) have been shown to be equivalent in ADT networks \cite{adt1}\cite{fragouli}\cite{goemans}. We now show that 1) and 3) are equivalent. Assume that there exists an assignment such that $\det(M) \ne 0$ in $\mathbb{F}_q$. Then, the system matrix $M$ is invertible; thus, there exists $M^{-1}$ such that $\m{X}(S) = \m{Z}M^{-1}$, and a connection of rate $R(c) = |\m{X}(S, T)|$ is established. Conversely, if connection $c$ is feasible, there exists a solution to the ADT network $G$ that achieves a rate of $R(c)$. When using this ADT network solution, the destination $T$ is able to reproduce $\m{X}(S, T)$; thus the resulting system matrix is an identity matrix, $M=I$. Therefore, $M$ is invertible.
\end{proof}

\begin{corollary}[Random Coding for Unicast]Consider an ADT network problem with a single connection $c = (S, T, \m{X}(S, T))$ of rate $R(c) = |\m{X}(S, T)| \leq mincut(S,T)$. Then, random linear network coding, where some or all code  variables $\alpha_{(i, e_j)}$, $\epsilon_{(e_i, (T_j, k))}$, and $\beta_{(e_i, e_j)}$ are chosen independently and uniformly over all elements of $\mathbb{F}_q$, guarantees decodability at destination $T$ with high probability at least $(1-\frac{1}{q})^{\eta}$, where $\eta$ is the number of links carrying random combinations of the source processes.
\end{corollary}
\begin{proof}From Theorem \ref{thm:mincut_maxflow}, there exists an assignment of $\alpha_{(i, e_j)}$, $\epsilon_{(e_i, (T_j, k))}$, and $\beta_{(e_i, e_j)}$ such that $\det(M)\ne 0$, which gives a capacity-achieving network code for the given $(G, \m{C})$. Thus, this connection $c$ is feasible for the given network. Reference \cite{rlc} proves that random linear network coding is capacity-achieving and guarantees decodability with high probability $(1-\frac{1}{q})^{\eta}$ for such a feasible unicast connection $c$.
\end{proof}

\begin{theorem}[Single Multicast Theorem] Given an acyclic network $G$ and connections $\m{C}= \{(S, T_1, \m{X}(S)),$ $(S, T_2,$ $\m{X}(S)),$ $..., (S, T_N, \m{X}(S))\}$, $(G, \m{C})$ is solvable if and only if $mincut(S, T_i) \geq |\m{X}(S)|$ for all $i$. \label{thm:singlemulticast}
\end{theorem}

\begin{proof} If $(G, \m{C})$ is solvable, then $mincut(S, T_i) \geq |\m{X}(S)|$. Therefore, we only have to show the converse. Assume $mincut(S, T_i) \geq |\m{X}(S)|$ for all $i \in [1, N]$. The system matrix $M = \{M_i\}$ is a concatenation of $|\m{X}(S)| \times |\m{X}(S)|$ matrices where $\m{Z}(T_i) = \m{X}(S) M_i$, as shown in Figure \ref{fig:systemmatrix}. We can write $M = [M_1, M_2, ..., M_N] = A (I-F)^{-1} B^T = A(I-F)^{-1}[B_1, B_2, ..., B_N]$. Thus, $M_i = A(I-F)^{-1}B_i$. Note that $A$ and $B_i$'s do not substantially contribute to the system matrix $M_i$ since $A$ and $B_i$ only perform linear encoding and decoding at the source and destinations, respectively.

By Theorem \ref{thm:mincut_maxflow}, there exists an assignment of $\alpha_{(i, e_j)}$, $\epsilon_{(e_i, (T_j, k))}$, and $\beta_{(e_i, e_j)}$ such that each individual system submatrix $M_i$ is invertible, \ie $\det{(M_i)} \ne 0$. However, an assignment that makes $\det{(M_i)} \ne 0$ may lead to $\det{(M_j)} = 0$ for $i\ne j$. Thus, we need to show that it is possible to achieve \emph{simultaneously} $\det{(M_i)} \ne 0$ for all $i$ (equivalently $\prod_{i} \det{(M_i)} \ne 0$). By \cite{rlc}, we know that if the field size is larger than the number of receivers ($q > N$), then there exists an assignment of $\alpha_{(i, e_j)}$, $\epsilon_{(e_i, (T_j, k))}$, and $\beta_{(e_i, e_j)}$ such that $\det{(M_i)} \ne 0$ for all $i$.
\end{proof}

\begin{corollary}[Random Coding for Multicast] Consider an ADT network problem with a single multicast connection $\m{C}= \{(S,T_1, \m{X}(S)), (S, T_2, \m{X}(S)), ..., (S, T_N, \m{X}(S))\}$ with $mincut(S,T_i) \geq |\m{X}(S)|$ for all $i$.  Then, random linear network coding, where some or all code variables $\alpha_{(i, e_j)}$, $\epsilon_{(e_i, (T_j, k))}$, and $\beta_{(e_i, e_j)}$ are chosen independently and uniformly over all elements of $\mathbb{F}_q$, guarantees decodability at destination $T_i$ for all $i$ simultaneously with high probability at least $(1-\frac{N}{q})^{\eta}$, where $\eta$ is the number of links carrying random combinations of the source processes; thus, $\eta \leq |\m{E}|$.\label{thm:coding_multicast}
\end{corollary}
\begin{proof}Given that the multicast connection is feasible (which is true by Theorem \ref{thm:singlemulticast}), reference \cite{rlc} shows that random linear network coding achieves the capacity for multicast connections, and allows all destination supernodes to decode the source processes $\m{X}(S)$ with high probability of at least $(1-\frac{N}{q})^{\eta}$.
\end{proof}

Theorem \ref{thm:mincut_maxflow} and Theorem \ref{thm:singlemulticast} provide an alternate proof of sufficiency of linear operations for unicast and multicast in ADT networks, which was first shown in \cite{adt1}.

\section{Extensions to other connections}\label{sec:general}

In this section, we extend the ADT network results to a more general set of traffic requirements. We use the algebraic formulation and the results from \cite{algebraic} to characterize the feasibility conditions for a given problem $(G, \m{C})$.

\subsection{Multiple Multicast}

\begin{theorem}[Multiple Multicast Theorem] Given a network $G$ and a set of connections $\m{C} = \{(S_i, T_j, \m{X}(S_i))\ |\ S_i \in \m{S}, T_j \in \m{T}\}$, $(G, \m{C})$ is solvable if and only if Min-cut Max-flow bound is satisfied for any cut that separates the source supernodes $\m{S}$ and a destination $T_j$, for all $T_j \in \m{T}$.\label{thm:multiplemulticast}
\end{theorem}
\begin{proof} We first introduce a super-source $S$ with $|O(S)| = \sum_{S_i \in \m{S}} |O(S_i)|$, and connect each $e'_j \in O(S)$ to an input of $S_i$ such that $e_j \in O(S_i)$ as shown in Figure \ref{fig:supersource}. Then, we apply Theorem \ref{thm:singlemulticast}, which proves the statement.
\end{proof}

\begin{corollary}[Random Coding for Multiple Multicast] Consider an ADT network problem with multiple multicast connections $\m{C}= \{(S_i,T_j, \m{X}(S_i)) | S_i \in \m{S}, T_j \in \m{T} \}$ with $mincut(\m{S},T_j) \geq \sum_{i} |\m{X}(S_i)|$ for all $i$.  Then, random linear network coding, where some or all code variables $\alpha_{(i, e_j)}$, $\epsilon_{(e_i, (T_j, k))}$, and $\beta_{(e_i, e_j)}$ are chosen independently and uniformly over all elements of $\mathbb{F}_q$, guarantees decodability at destination $T_i$ for all $i$ simultaneously with high probability at least $(1-\frac{N}{q})^{\eta}$, where $\eta$ is the number of links carrying random combinations of the source processes; thus, $\eta \leq |\m{E}|$.\label{thm:coding_multi-multicast}
\end{corollary}

The optimality of random coding in Corollary \ref{thm:coding_multi-multicast} comes from the fact that we allow coding across multicast connections $(S_i, T_j, \m{X}(S_i))$'s -- \ie, the source supernodes and the intermediate supernodes can randomly and uniformly select the coding coefficients. Thus, intermediate nodes within the network do not distinguish the flow from source $S_i$ from that of $S_j$, and are allowed to encode them together randomly.

\subsection{Disjoint Multicast}\label{sec:disjoint}

\begin{theorem}[Disjoint Multicast Theorem] Given an acyclic network $G$ with a set of connections $\m{C} = $ $ \{(S, T_i, $ $ \m{X}(S, T_i))$ $\ | \ i = 1,2, ...,K\}$ is called a \emph{disjoint multicast} if $\m{X}(S, T_i) \cap \m{X}(S, T_j) = \emptyset$ for all $i\ne j$. Then, $(G, \m{C})$ is solvable if and only if the min-cut between $S$ and any subset of destinations $\m{T}' \subseteq \m{T}$ is at least $\sum_{T_i \in \m{T'}} |\m{X}(S, T_i)|$, \ie $mincut(S, \m{T}') \geq \sum_{T_i \in \m{T'}} |\m{X}(S, T_i)|$ for any $\m{T}' \subseteq \m{T}$.\label{thm:disjoint}
\end{theorem}%
\begin{proof}
Create a super-destination supernode $T$ with $|I(T)| = \sum_{i=1}^K |I(T_i)|$, and an edge $(e, e')$ from $e\in O(T_i)$, $i\in [1, K]$ to $e'\in I(T)$, as in Figure \ref{fig:superdestin}. This converts the problem of disjoint multicast to a single-source $S$, single-destination $T$ problem with rate $\m{X}(S, T) = \sum_{T' \in \m{T}} |\m{X}(S, T)|$. The $mincut(S, T) \geq |\m{X}(S, T)|$; so, Theorem \ref{thm:mincut_maxflow} applies. Thus, it is possible to achieve a communication of rate $\m{X}(S,T)$ between $S$ and $T$. Now, we have to guarantee that the receiver $T_i$ is able to receive the exact subset of processes $\m{X}(S,T_i)$. Since the system matrix to $T$ is full rank, it is possible to carefully choose the encoding matrix $A$ such that the system matrix $M$ at super-destination supernode $T$ is an identity matrix. This implies that for each edge from the output ports of $T_i$ (for all $i$) to input ports of $T$ is carrying a distinct symbol, disjoint from all the other symbols carried by those  edges from output ports of $T_j$, for all $i \ne j$. Thus, by appropriately permuting the symbols at the source, $S$ can deliver the desired processes to the intended $T_i$ as shown in Figure \ref{fig:superdestin}.
\end{proof}

\begin{figure}
\begin{center}
\includegraphics[width=0.45\textwidth]{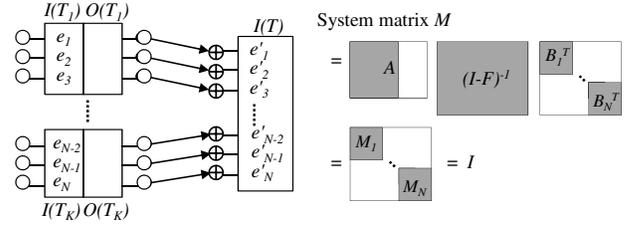}
\end{center}\vspace*{-.1cm}\caption{Disjoint multicast problem can be converted into a single destination problem by adding a super-destination $T$. The system matrix $M$ for the disjoint multicast problem is shown as well. Note that unlike the multicast problem in Figure \ref{fig:multicast} where $M = [M_1, M_2, ...,M_N]$, the system matrix $M$ is a diagonal concatenation of $M_i$'s.
}\label{fig:superdestin}
\end{figure}

Random linear network coding with a minor modification achieves the capacity for disjoint multicast. We note that only the source's encoding matrix $A$ needs to be modified. The intermediate supernodes can randomly and uniformly select coding coefficients $\epsilon_{(e_i, (T_j, k))}$ and $\beta_{(e_i, e_j)}$ over all elements of $\mathbb{F}_q$. Once these coding coefficients at the intermediate supernodes are selected, $S$ carefully chooses the encoding matrix $A$ such that the system matrix corresponding to the receivers of the disjoint multicast is an identity matrix, as shown in Figure \ref{fig:superdestin}. To be more precise, when $\epsilon_{(e_i, (T_j, k))}$ and $\beta_{(e_i, e_j)}$ are randomly selected over elements of $\mathbb{F}_q$, with high probability, $(I-F)^{-1}B^T$ is full rank. Thus, there exists a matrix $A$ such that $A (I-F)^{-1}B^T$ is an identity matrix $I$. Note that $A(I-F)^{-1}B^T$ does not need to be an identity matrix -- it only needs to have a diagonal structure as shown in Figure \ref{fig:superdestin}; however, being an identity matrix is sufficient for proof of optimality.

We note another subtlety here. Theorem \ref{thm:disjoint} holds precisely because we allow the intermediate nodes to code across all source processes, even they are destined for different receivers. This takes advantage of the fact that the single source can cleverly pre-code the data.

\subsection{Two-level Multicast}\label{sec:twolevel}

\begin{theorem}[Two-level Multicast Theorem] Given an acyclic network $G$ with a set of connections $\m{C} =\m{C}_{d} \cup \m{C}_{m}$ where $\m{C}_{d} = \{(S, T_i, \m{X}(S, T_i)) | \m{X}(S, T_i) \cap \m{X}(S, T_j) = \emptyset,$ $ i\ne j$, $i, j \in [1,K]\}$ is a disjoint multicast connection, and $\m{C}_{m} = \{(S, T_i, \m{X}(S))\ | \ i \in [K+1, N]\}$ is a single source multicast connection. Then, $(G, \m{C})$ is solvable if and only if the min-cut between $S$ and any $\m{T}' \subseteq \{T_1, ..., T_K\}$ is at least $\sum_{T_i \in \m{T}'} |\m{X}(S, T_i)|$, and the min-cut between $S$ and $T_j$ is at least $|\m{X}(S)|$ for $j \in [K+1, N]$.\label{thm:twolevel}
\end{theorem}
\begin{proof}
We create a super-destination $T$ for the disjoint multicast destinations as in the proof for Theorem \ref{thm:disjoint}. Then, we have a single multicast problem with receivers $T$ and $\{T_i | i\in [K+1, N]\}$. Therefore, Theorem \ref{thm:singlemulticast} applies. By choosing the appropriate matrix $A$, $S$ can satisfy both the disjoint multicast and the single multicast requirements, as shown in Figure \ref{fig:twolevel}.
\end{proof}

\begin{figure}[tbp]
\begin{center}
\includegraphics[width=0.49\textwidth]{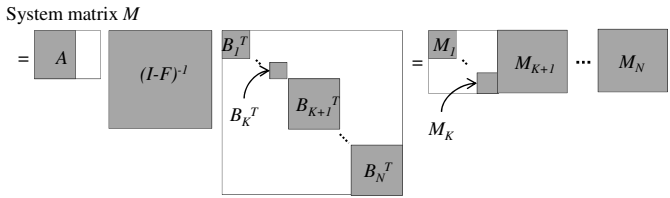}
\end{center}\vspace*{-.1cm}\caption{The system matrix $M$ for the two-level multicast problem. The structure of the system matrix $M$ is a ``concatenation'' of the disjoint multicast problem (Figure \ref{fig:superdestin}) and the single multicast problem (Figure \ref{fig:multicast}).}\label{fig:twolevel}
\end{figure}

As in the disjoint multicast case, random linear network coding with a minor modification at the source achieves the capacity for two-level multicast. Note that, receivers $T_i$, $i \in [K+1, N]$ are of no concern --  the source $S$ can randomly choose coding coefficients $\alpha_{(i, e_j)}$ to achieve a full-rank system matrix $M_{i}$. Thus, $S$ needs to carefully choose the encoding matrix $A$ to satisfy the disjoint multicast constraint, which can be done as shown in Section \ref{sec:disjoint}.

Theorem \ref{thm:twolevel} does not extend to a three-level multicast. Three-level multicast, in its simplest form, consists of connections $\{(S, T_i, \m{X}(S, T_i))|\ i \in [1,3]\}$ where $\m{X}(S, T_1) \subset \m{X}(S, T_2) \subset \m{X}(S, T_3)$.

\subsection{General Connection Set}

In the theorem below, we present sufficient conditions for solvability of a general connection set. This theorem does not provide necessary conditions, as shown in \cite{insufficiency}.

\begin{theorem}[Generalized Min-cut Max-flow Theorem] Given an acyclic network $G$ with a connection set $\m{C}$, let $M = \{M_{i,j}\}$ where $M_{i,j}$ is the system matrix for source processes $\m{X}(S_i)$ to destination processes $\m{Z}(T_j)$. Then, $(G, \m{C})$ is solvable if there exists an assignment of $\alpha_{(i, e_j)}$, $\epsilon_{(e_i, (T_j, k))}$, and $\beta_{(e_i, e_j)}$ such that
\begin{enumerate}
\item $M_{i,j} = 0$ for all $(S_i, T_j, \m{X}(S_i, T_j)) \notin \m{C}$,
\item Let $(S_{\sigma(i)}, T_j, \m{X}(S_{\sigma(i)}, T_j)) \in \m{C}$ for $i \in [1, K(j)]$. Thus, this is the set of connections with $T_j$ as a receiver. Then, $[M_{\sigma(1),j}^T, M_{\sigma(2), j}^T, ...,$ $M_{\sigma(K_j), j}^T]$ is a $|\m{Z}(T_j))| \times |\m{Z}(T_j)|$ is a \emph{nonsingular}
    system matrix.
\end{enumerate}
\end{theorem}
\begin{proof}
Note that $[M_{\sigma(1),j}^T, M_{\sigma(2), j}^T, ...,$ $M_{\sigma(K_j), j}^T]$  is a system matrix for source processes $\m{X}(S_{\sigma(i)})$, $i \in [1, K(j)]$, to destination processes $\m{Z}(T_j)$.

Condition 2) states the Min-cut Max-flow condition; thus, is necessary to establish the connections. Condition 1) states that the destination supernode $T_j$ should be able to distinguish the information it is intended to receive from the information that may have been mixed into the flow it receives. These two conditions are sufficient to establish all connections in $\m{C}$. The proof is similar to that of Theorem 6 in \cite{algebraic}.
\end{proof}


\section{Network with Random Erasures}\label{sec:robust}

We consider the algebraic ADT problem where links may fail randomly, and cause erasures. Wireless networks are stochastic in nature, and random erasures occur dynamically over time. However, the original ADT network models noise deterministically with parallel noise-free bit-pipes. As a result, the min-cut (Definition \ref{def:mincut}) and the network code \cite{fragouli}\cite{goemans}\cite{edmund}, which depend on the hard-coded representation of noise, have to be recomputed every time the network changes.

We show that the algebraic framework for the ADT network is robust against random erasures and failures. First, we show that for some set of link failures, the network code remains successful. This translates to whether the system matrix $M$ preserves its full rank even after a subset of variables $\alpha_{(i, e_j)}, \epsilon_{(e_i, (D_j, k))}$, and $\beta_{(e_i, e_j)}$ associated with the failed links is set to zero. Second, we show that the specific instance of the system matrix $M$ and its rank are not as important as the \emph{average} $\text{rank}(M)$ when computing the time average min-cut. Note that the original min-cut definition (Definition \ref{def:mincut}) requires an optimization over an exponential number of cuts for every time step to find the average min-cut. 
We shall use the results from \cite{reliable} to show that random linear network coding achieves the time-average min-cut; thus, is capacity-achieving.

We assume that any link within the network may fail. Given an ADT network $G$ and a set of link failures $f$, $G_f$ represents the network $G$ experiencing failures $f$. This can be achieved by deleting the failing links from $G$, which is equivalent to setting the coding variables in $B(f)$ to zero, where $B(f)$ is the set of coding variables associated with the failing links. We denote $M$ be the system matrix for network $G$. Let $M_{f}$ be the system matrix for the network $G_f$. We do not assume that the link failures are static; thus, we can consider a static link failure patterns, a distribution over link failures patterns, or a sequence $f_1, f_2, f_3 ...$ of link failures.

\subsection{Robust against Random Erasures}\label{sec:erasures}

Given an ADT network problem $(G, \m{C})$, let $\m{F}$ be the set of \emph{all} link failures such that, for any $f\in \m{F}$, the problem $(G_f, \m{C})$ is solvable. The solvability of a given $(G_f, \m{C})$ can be verified using resulting in Sections \ref{sec:singlesource} and \ref{sec:general}.  We are interested in static solutions, where the network is oblivious of $f$.
In other words, we are interested in finding the set of link failures such that the network code is still successful in delivering the source processes to the destinations.
For a multicast connection, we show the following surprising result.

\begin{theorem}[Static Solution for Random Erasures]\label{thm:static}
Given an ADT network problem $(G, \m{C})$ with a multicast connection $\m{C} = \{(S, T_1, \m{X}(S)),$ $(S, T_2, \m{X}(S)), ..., (S, T_N, \m{X}(S))\}$, there exists a \emph{static} solution to the problem $(G_f, \m{C})$ for all $f\in \m{F}$. In other words, there exists a \emph{fixed} network code that achieves the multicast rate despite any failures $f \in \m{F}$.
\end{theorem}
\begin{proof}
By Theorem \ref{thm:singlemulticast}, we know that for any given $f\in \m{F}$, the problem $(G_f, \m{C})$ is solvable; thus, there exists a code $\det{(M_f)}\ne 0$. Now, we need to show that there exists a code such that $\det{(M_f)}\ne 0$ for all $f\in \m{F}$ simultaneously. This is equivalent to finding a non-zero solution to the following polynomial: $\prod_{f\in \m{F}} \det{(M_f)} \ne 0$. Reference \cite{rlc} showed that if the field size is large enough ($q > |\m{F}||\m{T}| = |\m{F}|N$), then there exists an assignment of $\alpha_{(i, e_j)}, \epsilon_{(e_i, (D_j, k))}$, and $\beta_{(e_i, e_j)}$ such that $\det{(M_f)} \ne 0$ for all $ f \in \m{F}$.
\end{proof}

\begin{corollary}[Random Coding against Random Erasures]\label{thm:static_coding} Consider an ADT network problem with a multicast connection $\m{C} = \{(S, T_1, \m{X}(S)),$ $(S, T_2, \m{X}(S)), ..., (S, T_N, \m{X}(S))\}$, which is solvable under link failures $f$, for all $f\in \m{F}$. Then, random linear network coding, where some or all code variables $\alpha_{(i, e_j)}, \epsilon_{(e_i, (D_j, k))}$, and $\beta_{(e_i, e_j)}$ are chosen independently and uniformly over all elements of $\mathbb{F}_q$ guarantees decodability at destination supernodes $T_i$ for all $i$ simultaneously and remains successful regardless of the failure pattern $f \in \m{F}$ with high probability at least $(1- \frac{N|\m{F}|}{q})^\eta$, where $\eta$ is the number of links carrying random combinations of the source processes.
\end{corollary}
\begin{proof}
Given a multicast connection that is feasible under any link failures $f \in \m{F}$, \cite{rlc} showed that random linear network coding achieves the capacity for multicast connections, and is robust against any link failures $f \in \m{F}$ with high probability $(1 - \frac{N |\m{F}|}{q})^\eta$.
\end{proof}

We note that it is unclear whether this can be extended to the non-multicast connections, as noted in \cite{algebraic}. Reference \cite{algebraic} shows a simple example network in which no static solution is available for a set of feasible failure patterns.

\subsection{Time-average Min-cut}\label{sec:average}

In this section, we study the time-average behavior of the ADT network, given random erasures. We use techniques from \cite{reliable}, which studies reliable communication over lossy networks with network coding.

Consider an ADT network $G$. In order to study the time-average steady state behavior, we introduce erasure distributions. Let $\m{F}'$ be a set of link failure patterns in $G$. A set of link failures $f \in \m{F'}$ may occur with probability $p_f$.

\begin{theorem}[Min-cut for Time-varying Network]\label{thm:average_mincut}
Assume an ADT network $G$ in which link failure pattern $f \in \m{F}'$ occurs with probability $p_f$. Then, the average min-cut between two supernodes $S$ and $T$ in $G$, $mincut_{\m{F}'}(S,T)$ is
\[
mincut_{\m{F}'} (S,T) = \sum_{f\in \m{F}'} p_f \left(\max_{\alpha_{(i, e)}, \beta_{(e', e)}, \epsilon_{(e', i)}} \text{rank}(M_f)\right).
\]
\end{theorem}
\begin{proof}
By Theorem \ref{thm:mincut}, we know that at any given time instance with failure pattern $f$, the min-cut between $S$ and $T$ is given by $\max_{\alpha_{(i, e)}, \beta_{(e', e)}, \epsilon_{(e', i)}} \text{rank}(M_f)$. Then, the above statement follows naturally by taking a time average of the min-cut values between $S$ and $T$.
\end{proof}

The key difference between Theorems \ref{thm:static} and \ref{thm:average_mincut} is that in Theorem \ref{thm:static}, any $f\in \m{F}$ may change the network topology as well as min-cut but $mincut(S,T) \geq |\m{X}(S)|$ holds for all $f\in \m{F}$ -- \ie $(G_f, \m{C})$ is assumed to be solvable. However, in Theorem \ref{thm:average_mincut}, we make no assumption about the connection as we are evaluating the average value of the min-cut.

Unlike the case of static ADT networks, with random erasures, it is necessary to maintain a queue at each supernode in the ADT network. This is because, if a link fails when a supernode has data to transmit on it, then it will have to wait until the link recovers. In addition, a transmitting supernode needs to be able to learn whether a packet has been received by the next hop supernode, and whether it was innovative -- this can be achieved using channel estimation, feedback and/or redundancy. In the original ADT network, the issue of feedback was removed by assuming that the links are noiseless bit-pipes. We present the following corollaries under these assumptions.

\begin{corollary}[Multicast in Time-varying Network]\label{thm:average_multicast} Consider an ADT network $G$ and a multicast connection $\m{C} = \{(S, T_1, \m{X}(S)), ..., (S, T_N, \m{X}(S))\}$. Assume that failures occur where failure patten $f \in \m{F}'$ occurs with probability $p_f$. Then, the multicast connection is feasible if and only if $mincut_{\m{F}'}(S,T_i) \geq |\m{X}(S)|$ for all $i$.
\end{corollary}
\begin{proof} Reference \cite{reliable} shows that the multicast connection is feasible if and only $mincut_{\m{F}'}(S,T_i) \geq |\m{X}(S)|$ for all $i$. The proof in \cite{reliable} relies on the fact that every supernode behaves like a stable $M/M/1$ queuing system in steady-state, and thus, the queues (or the number of innovative packets to be sent to the next hop supernode) has a finite mean if the network is run for sufficiently long period of time.
\end{proof}
\begin{corollary}[Random Coding for Time-varying Network] Consider $(G, \m{C})$ where $\m{C}$ is a multicast connection. Assume failure pattern $f \in \m{F}'$ occurs with probability $p_f$. Then, random linear coding, where some or all code variables $\alpha_{(i, e_j)}, \beta_{(e_i, e_j)}, \epsilon_{(e_i, (D_j, k))}$ are chosen over all elements of $\mathbb{F}_q$ guarantees decodability at destination nodes $T_i$ for all $i$ simultaneously with arbitrary small error probability.
\end{corollary}
\begin{proof} This is a direct consequence of Corollary \ref{thm:average_multicast} and results in \cite{rlc}\cite{reliable}.
\end{proof}

\section{Network with Cycles}\label{sec:delay}

\begin{figure*}[tbp]
\[\scriptsize
\left(
\begin{array}{cccccccccccc}
1 & 0 & D & 0 & 0 & D & D^2\beta_{(e_3, e_7)} & 0 & D^2 \beta_{(e_6, e_9)} & D^2\beta_{(e_6, e_{10})} & D^3 \beta_{(e_6, e_9)} & D^3\beta_{(e_3, e_7)} + D^3 \beta_{(e_6, e_{10})}\\
0 & 1 & 0 & D & 0 & 0 & D^2 \beta_{(e_4, e_7)} & 0 & 0 & 0 & 0 & D^3 \beta_{(e_4, e_7)}\\
0 & 0 & 1 & 0 & 0 & 0 & D\beta_{(e_3, e_7)} &  0 & 0 & 0 & 0 & D^2 \beta_{(e_3, e_7)}\\
0 & 0 & 0 & 1 & 0 & 0 & D \beta_{(e_4, e_7)} & 0 & 0 & 0 & 0 & D^2 \beta_{(e_4, e_7)}\\
0 & 0 & 0 & 0 & 1 & 0 & 0 & 0 & 0 & 0 & 0 & 0\\
0 & 0 & 0 & 0 & 0 & 1 & 0 & 0 & D\beta_{(e_6, e_{9})} & D\beta_{(e_6, e_{10})}& D^2 \beta_{(e_6, e_9)} & D^2 \beta_{(e_6, e_{10})}\\
0 & 0 & 0 & 0 & 0 & 0 & 1 & 0 & 0 & 0 & 0 & D\\
0 & 0 & 0 & 0 & 0 & 0 & 0 & 1 & 0 & 0 & 0 & 0\\
0 & 0 & 0 & 0 & 0 & 0 & 0 & 0 & 1 & 0 & D & 0\\
0 & 0 & 0 & 0 & 0 & 0 & 0 & 0 & 0 & 1 & 0 & D\\
0 & 0 & 0 & 0 & 0 & 0 & 0 & 0 & 0 & 0 & 1 & 0\\
0 & 0 & 0 & 0 & 0 & 0 & 0 & 0 & 0 & 0 & 0 & 1\\
\end{array}
\right)
\]\vspace*{-.1cm}\caption{$12 \times 12$ matrix $(I-DF)^{-1}$ for network in Figure \ref{fig:network}. The matrix $F$ can be found in Figure \ref{fig:F}.}\label{fig:DF}
\end{figure*}

ADT networks are acyclic, with links directed from the source supernodes to the destination supernodes. However, wireless networks intrinsically have cycles as wireless links are bi-directional by nature. In this section, we extend the ADT network model to networks with cycles. In order to incorporate cycles, we need to introduce the notion of time -- since, without the notion of time, the network with cycles may not be casual. To do so, we introduce delay on the links. As in \cite{algebraic}, we model each link to have the same delay, and express the network random processes in the delay variable $D$.

We define $X_t(S, i)$ and $Z_t(T, j)$ to be the $i$-th and $j$-th binary random process generated at source $S$ and received at destination $T$ at time $t$, for $t = 1, 2,...$. We define $Y_t(e)$ to be the process on edge $e$ at time $t = 1, 2, ...$, respectively. We express the source processes as a power series in $D$, $\m{X}(S, D) = [X(S, 1, D), X(S, 2, D), ..., X(S, \mu(S), D)]$ where $X(S, i, D) = \sum_{t=0}^\infty X_t(S, i)D^t$. Similarly, we express the destination random processes $\m{Z}(T, D) = [Z(T, 1, D),$ $...,Z(T, \nu(Z), D)]$ where $Z(T, i, D) = \sum_{t=0}^\infty Z_t(T, i)D^t$. In addition, we express the edge random processes as $Y_t (e, D) = \sum_{t=0}^\infty Y_t(e)D^t$. Then, we can rewrite Equations (\ref{eq:y}) and (\ref{eq:y-s}) as
\begin{equation*}\label{eq:time_y-s}
Y_{t+1}(e) =\sum_{e'\in I(V)} \beta_{(e', e)} Y_t(e') + \sum_{X_t(S,i) \in \m{X}(S)} \alpha_{(i, e)} X_t(S, i).
\end{equation*}
Furthermore, the output processes $Z_t(T,i)$ can be rewritten as
\begin{align*}\label{eq:time_z}
Z_{t+1}(T, i)  =  \sum_{e' \in I(T)}\epsilon_{e',(T, i)} Y_{t}(e').
\end{align*}
%
Using this formulation, we can extend the results from \cite{algebraic} to ADT networks with cycles. We show that a system matrix $M(D)$ captures the input-output relationships of the ADT networks with delay and/or cycles.

\begin{theorem}\label{thm:delay}
Given a network $G = (\m{V}, \m{E})$, let $A(D)$, $B(D)$, and $F$ be the encoding, decoding, and adjacency matrices, as defined here:
\begin{align*}
A_{i,j} &=
\begin{cases}
\alpha_{(i, e_j)}(D) &\text{if $e_j \in O(S)$ and $X(S, i) \in \m{X}(S)$,}\\
0 &\text{otherwise.}
\end{cases}\\
B_{i, (T_j, k)} &=
\begin{cases}
\epsilon_{(e_i, (T_j, k))}(D) &\text{if $e_i \in I(T_j)$, $Z(T_j, k) \in \m{Z}(T_j)$,}\\
0 &\text{otherwise}.
\end{cases}
\end{align*}
and $F$ as in Equation (\ref{eq:F}). The variables $\alpha_{(i, e_j)} (D)$ and $\epsilon_{(e_i, (T_j, k))}(D)$ can either be constants or rational functions in $D$. Then, the system matrix of the ADT network with delay (and thus, may include cycles) is given by
\begin{equation}
M (D) = A(D) \cdot (I - DF)^{-1} \cdot B(D)^T.
\end{equation}
\end{theorem}
\begin{proof}The proof for this is similar to that of Theorem \ref{thm:m}.
\end{proof}

Similar to Section \ref{sec:algebraic}, $(I-DF)^{-1}$ represents the impulse response of the network with delay. This is because the series $I + DF + D^2F^2 + D^3F^3 + ...$ represents the connectivity of the network while taking delay into account. For example, $F^k$ has a non-zero entry if there exists a path of length $k$ between two port. Now, since we want to represent the time associated with traversing from port $e_i$ to $e_j$, we use $D^kF^k$, where $D^k$ signifies that the path is of length $k$. Thus, $(I-DF)^{-1} = I + DF + D^2F^2 + D^3F^3 + ...$ is the impulse response of the network with delay. An example of $(I-DF)^{-1}$ for the example network in Figure \ref{fig:newnetwork} is shown in Figure \ref{fig:DF}.

Using the system matrix $M(D)$ from Theorem \ref{thm:delay}, we can extend Theorem \ref{thm:mincut_maxflow}, Theorem \ref{thm:singlemulticast}, Theorem \ref{thm:multiplemulticast}, Theorem \ref{thm:disjoint}, and Theorem \ref{thm:twolevel} to ADT networks with cycles/delay. However, there is a minor technical change. We now operate in a different field -- instead of having coding coefficients from the finite field $\mathbb{F}_q$, the coding coefficients $\alpha_{(i, e_j)}(D)$ and $\epsilon_{((e_i, (T_j, k)))}(D)$ are now from $\mathbb{F}_q(D)$, the field of rational functions of $D$. We shall not discuss the proofs in detail; however, this is a direct application of the results in \cite{algebraic}.

\section{Code Construction for Multicast Connection}\label{sec:code}
As presented in Sections \ref{sec:singlesource}-\ref{sec:delay}, random linear network coding, which is completely distributed, requires randomization. As a result, random linear network coding achieves the capacity for the ADT network model \emph{with high probability}. Similarly, \cite{jaggi} introduced a distributed binary-vector network code, called \emph{permute-and-add}, in which each node randomly and uniformly selects a permutation matrix from all such matrices for its coding operation. In \cite{jaggi}, they show that permute-and-add is still optimal for multicast connections -- \ie achieves the capacity with high probability. Thus, network codes in $\mathbb{F}_q$, $q \geq 2$, can be converted to vector code in binary field, $(\mathbb{F}_2)^{\log_2(q)}$, without loss in performance. As a result, permute-and-add can be applied to the ADT network model. However, it is important to note that a randomized, distributed approach does not \emph{guarantee decodability} (with probability 1).

In this section, we propose an efficient code construction for multicast connection in ADT network, which \emph{guarantees all destination supernodes to decode if the connection is feasible}. Furthermore, we only require that there be some \emph{local} coordination among neighboring supernodes (more precisely, among supernodes within a layer), as we shall discuss in Section \ref{code_cons_input}.

Given an acyclic ADT network $G = (\m{V}, \m{E})$ and a multicast connection $\m{C} = \{(S, T_i, \m{X}(S)) | T_i \in\m{T}\}$, we consider the problem of efficiently finding an assignment for $\alpha_{(i, e_j)}, \epsilon_{(e_i, (D_j, k))}$, and $\beta_{(e_i, e_j)}$ such that the system matrix $M$ is invertible in $\mathbb{F}_q$. Note that we assume that $(G, \m{C})$ is solvable -- \ie $mincut(S, T_i) \geq R$ for all $i$, where $R = |\m{X}(S)|$ is the multicast rate. We shall consider $R = \min_{T_i \in \m{T}} mincut(S, T_i)$.

For the given multicast rate $R$, we define the set
\begin{equation}
\label{vhmin} \m{V}(R)\defined\{V \in \m{V}|mincut(S, V)\geq R\}.
\end{equation}
A property of our code construction is that all supernodes in $\m{V}(R)$ will be able to decode the data, including those that are not in the set of destination supernodes $\m{T}$. We also note that the code designer may be oblivious of the exact location of the nodes in $T$ or $\m{V}(R)$.

We assume that each supernode $V \in \m{V}$ contains $n$ input ports and $n$ output ports, where
\begin{equation}
n = \left\lceil \frac{1}{2} \log \max_{(V_i, V_j) \in \m{E}} SNR_{(V_i, V_j)}\right\rceil.
\end{equation}
Therefore, we assume that $|I(V)|=|O(V)|=n$. For ease of notation, we distinguish the input ports of supernode $V$ by
\begin{equation}
\label{gamma_def} I(V) \defined\{x^V_1,\cdots,x^V_n\},
\end{equation}
and the output ports by
\begin{equation}
\label{gamma_def_out}
O(V)\defined\{y^V_1,\cdots,y^V_n\}.
\end{equation}

The network is assumed to have $\lambda$ layers, where all links are from layer $l$ to layer $l+1$ for $l = 1, ..., \lambda - 1$. The source is at layer 1. We assume that there are at most $N_{layer}$ nodes at any layer $l \in [1, \lambda]$. Since the network is acyclic, we can arrange all the ports in a topological order. The input ports of a certain supernode always precede the output ports of the same supernode. In addition, we adopt the convention that ports of supernodes in layer $l$ precede all the ports of supernodes in layer $l+1$, for $l=1,\cdots,\lambda-1$. We make the assumption that within a single layer, the supernodes are ordered from top to bottom. Also, within each supernode, ports are arranged from top to bottom.

We denote $P(x^V_i)$ to be the set of output ports that have links incoming into the input port $x^V_i$ of a supernode $V$ in layer $l$. By assumption, $0 \leq |P(x^V_i)| \leq N_{layer}$ -- \ie there are at most $N_{layer}$ edges incoming to $x^V_i$ from output ports in layer $l-1$. We denote $C(y^V_j)$ to be the set of input ports that have links outgoing to the output port $y^V_j$ of a supernode $V$ in layer $l$. Note that $0 \leq |C(y^V_j)| \leq N_{layer}$ since $y^V_j$ may be adjacent to input ports of different supernodes in layer $l+1$, but may not be adjacent to more than one input port per supernode in layer $l+1$.


\begin{figure}[tbp]
  \begin{center}
    \leavevmode
    \epsfig{file=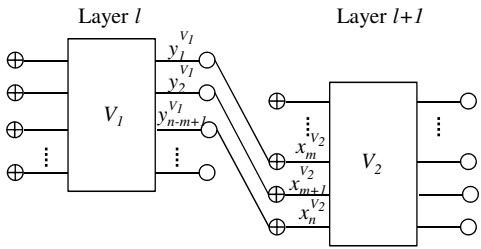,width=6.5cm}
        \caption{Supernodes $V_1$ and $V_2$.}
    \label{supernode2}
  \end{center}
\end{figure}

\subsection{Regular Sets and Virtual Sinks}
For the algorithm we describe in the following sections, we will use an invariant which will be maintained throughout the algorithm. Prior to defining this invariant, we need to introduce the concepts of regular sets and virtual sinks.

Consider supernodes at a certain layer $l$. We consider a set $W$ of $R$ ports, where for any supernode $V$ the set $W$ contains a subset of $V$'s output ports, or a subset of $V$'s input ports, but not both, or no ports of $V$. The set $W$ may contain ports of several supernodes. Now consider the following. If $W$ contains output ports of supernode $V$ then we connect each of the output ports to a ``virtual sink'' $T(W)$, which is a supernode consisting of its own ports. If $W$ contains $p$ input ports of supernode $V$, then we disconnect all the input ports of supernode $V$ that are not in $W$. We connect the $p$ upper output ports of $V$ to the virtual sink $T(W)$. The order in which the output ports are connected to $T(W)$ is not important. For consistency, however, we assume that the output ports of layer $l$ that are connected to sink $T(W)$, are connected to the input ports of $T(W)$ from top to bottom, each output port to a different input port of $T(W)$. See Figure \ref{regular_set} for an illustration.

\begin{definition}[\textbf{Regular set}]\label{den_regular} The set $W$ is said to be regular if $mincut(S, T(W)) = R$.
\end{definition}

\begin{figure}[tbp]
  \begin{center}
    \leavevmode
    \epsfig{file=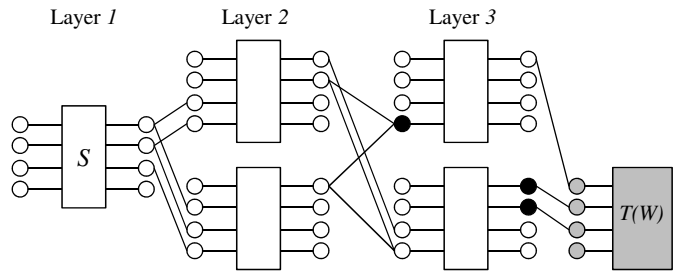,width=9cm}
        \caption{A set $W$ of input/output ports (in black) where $R=3$. The supernode and the ports shaded in grey represent the virtual sink $T(W)$. Note that $W$ is \emph{not} regular.}
    \label{regular_set}
  \end{center}
\end{figure}

In Figure \ref{regular_set}, we show an example of $W$ where $R = 3$, and a virtual sink $T(W)$. The set $W$ is not regular since $T(W)$ receives only rate of $2$.

The important property of regular sets, as we shall observe, is that there exists a network code such that the coding vectors of the regular sets are linearly independent. We shall exploit this property for our code construction.

\subsection{Overview of Coding Scheme}
We proceed through the ports in the topological order, and for each port we reach, we choose the coding coefficients, taken from
$\mathbb{F}_q$, where $q$ is the field size to be determined. Each port has a coding vector associated with it, where coding vector represents the coding coefficients used to generate the linear combination. Denote the coding vectors of the input ports of supernode $V$ by
\begin{equation}
\mathbf{x}^{I(V)}=\{\mathbf{x}^V_1,\cdots,\mathbf{x}^V_n\},
\end{equation}
where $\mathbf{x}^V_i \in {\mathbb F}_q^{R}$ is the coding vector associated with the input port $x^V_i$. Similarly, we denote the coding vectors of the output ports of supernode $V$ by
\begin{equation}
\mathbf{y}^{O(V)} = \{\mathbf{y}^V_1, \cdots, \mathbf{y}^V_n\},
\end{equation}
where $\mathbf{y}^V_j \in {\mathbb F}_q^{R}$ is the coding vector associated with the output port $y^V_j$. The coding vector $\mathbf{y}^V_j$ is
\begin{equation}
\label{yrn_def} {\bf y}^V_j=\sum_{i=1}^n \beta_{(x^V_i,y^V_j)}{\bf x}^V_i,\ 1\leq j\leq n,
\end{equation}
where $\beta_{(x^V_i,y^V_j)}\in {\mathbb F}_q$ are the supernode internal coding coefficients, as described in Section \ref{sec:algebraic}.

Although we introduce new notation for the coding vectors for convenience, these vectors can still be captured by the system matrix $M = A (I-F)^{-1}B^T$. Consider an input port $x^V_i$. Assume that $x^V_i$ is the $p$th port in the topological order. Then, the coding vector $\mathbf{x}^V_i$ at this input port is $p$th column of the matrix $A(I-F)^{-1}$. This is because what input port $x^V_i$ receives depends on two things: first, the sources encoding operations, represented by $A$; second, the network connectivity and coding operations performed within the network, represented by $(I-F)^{-1}$. The same argument applies to the output ports $y^V_j$. Thus, the columns of $A(I-F)^{-1}$ is equal to the coding vectors at the corresponding input/output ports in the network.

We refer to this coding operation in Equation (\ref{yrn_def}) as ``forward coding". Once the coding vector $\mathbf{y}^V_j$ of the output port is determined, we can multiply it by the coding coefficient $k_j$. We refer to this step as ``virtual coding". The ``virtual coding" can be incorporated into the ``forward coding". However, we separate the coding into two distinct phases for purposes of presentation.

For an input port $x^V_i$, let $\mathbf{y}_1,\cdots,\mathbf{y}_p$ be the coding vectors of the output ports in the set $P(x^V_i)$. Then, the coding vector of the input port $x^V_i$ is given by
\begin{equation}
\mathbf{x}^V_i=\sum_{j=1}^p k_j \mathbf{y}_j,
\end{equation}
where $k_{j}\in {\mathbb F}_q$ are the virtual coding coefficients. Note that \emph{only a single coefficient $k_j$ is chosen for all ports in $C(y_j)$}, which requires that the supernodes in the same layer coordinate locally when determining $k_j$'s.

\begin{definition}[\textbf{Cut of ports}] Consider the coding scheme, which assigns coding coefficients in a topological order, as mentioned above. Let $t$ be the time index. We denote $\m{C}_t = (\hat{S}, \hat{S}^c)$ to be the current cut of the algorithm, where $\hat{S}$ is the set of ports whose coding coefficients have been determined, and $\hat{S}^c$ the set of remaining ports. An output port $y^V_j$ is in $\hat{S}$ if the coding coefficient $\beta_{(x^V_i, y^V_j)}$ of its supernode $V$ have already been determined. An input port $x^V_i$ is in $\hat{S}$ if \emph{all} of the virtual coefficients $k_j$ of output ports in $P(x_i^V)$ have been determined. The input ports of the source are in $\hat{S}$.
\end{definition}

A cut of ports ${\cal C}_t$ is not necessarily a cut of supernodes. In a cut of ports, ports of the same supernode can be in two different parts of the cut $\m{C}_t$. We do, however, restrict ourselves to a specific type of cuts of ports -- all the input ports of a specific supernode are on the same side of the cut, and all the output ports are on the same side of the cut.

\begin{definition}[\textbf{Boundary set}] A subset of ports $\m{Q}_{\m{C}_t}$ is a boundary set if
\begin{align*}
{\cal Q}_{{\cal C}_t} = &\ \{\textit{Input ports in }\hat{S}\textit{
with outputs }\textit{in
}\hat{S}^c\}\ \cup\\
&\ \{\textit{Output ports in }\hat{S}\textit{ with edges
outgoing to }\hat{S}^c\}.
\end{align*}
\end{definition}

Since the topological order proceeds from top to bottom at each layer, if the boundary set contains the input ports of a certain supernode $V$ at layer $l$, then it will also contain the input ports of the supernodes that are above supernode $V$ at layer $l$. Similarly, if the boundary set contains the output ports of a certain supernode $V$ then it will also contain the output ports of the supernodes that are above it at layer $l$. Figure \ref{exm3} shows an example of a boundary set, ${\cal Q}_{{\cal C}_t}$.

\begin{figure}[tbp]
  \begin{center}
    \leavevmode
    \epsfig{file=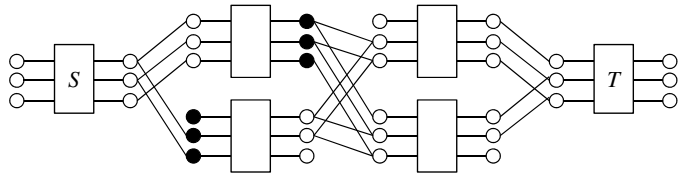,width=9cm}
        \caption{An example of ${\cal Q}_{{\cal C}_t}$ (in black).}
    \label{exm3}
  \end{center}
\end{figure}

\begin{lemma}
$R\leq n\leq |{\cal Q}_{{\cal C}_t}|\leq nN_{layer}$.
\end{lemma}
\begin{proof}
Consider a cut of supernodes $\Omega$ where source supernode $\Omega = \{S\}$ and $\Omega^c = \m{V}\setminus \{S\}$. Since $S$ has $n$ output ports, we conclude that $R \leq n$ since $R$ is upper bounded by the min-cut. By definition, the output ports of a supernode $V$ are restricted to be on the same side of the cut ${\cal C}_t$. The same is true for the input ports of a supernode. It follows that there are at least $n$ ports in ${\cal Q}_{{\cal C}_t}$ -- \ie $n\leq |{\cal Q}_{{\cal C}_t}|$. Since the network is assumed to be layered, and the maximal number of supernodes in a layer is $N_{layer}$, it follows from the definition of ${\cal Q}_{{\cal C}_t}$ that $|{\cal Q}_{{\cal
C}_t}|\leq nN_{layer}$.
\end{proof}

The code construction considers each subset of $R$ ports in ${\cal Q}_{{\cal C}_t}$. Some of the subsets we shall consider are regular sets. Define ${\cal L}_t$ by
\begin{eqnarray}
{\cal L}_t=\{\textit{Regular subsets of }{\cal Q}_{{\cal C}_t}
\textit{ of size } R\}.
\end{eqnarray}
The number of subsets in ${\cal L}_t$ is upper bounded by
\begin{equation}
|{\cal L}_t|\leq {|\m{Q}_{\m{C}_t}|\choose R} \leq {nN_{layer}\choose R}.
\end{equation}

\noindent\textbf{\textit{Code Invariant:}} \emph{Ensure that at each stage of the code
construction, each subset in ${\cal L}_t$ is associated with
linearly independent coding vectors.}

\begin{lemma}
Maintaining the invariant of the algorithm is sufficient to ensure the decodability of the code at rate $R$.
\end{lemma}
\begin{proof}For a (non-virtual) sink $T_j \in \m{T}$, let $W \subseteq O(T_j)$ be the $R$ upper output ports of $T_j$. By the definition of
$R$, we have $mincut(S, T_j) \geq R$. We connect the ports in $W$ to a virtual sink $T(W)$ with $R$ edges, where the $i$th output port of $T_j$ is connected to the $i$th input port of $T(W)$. It follows that  $R=mincut(S, T(W))$. Therefore, the set $W$ is regular.

The invariant ensures that the set $W$ will be associated with linearly independent coding vectors. It follows that $T_j$ will be able to decode the data, as required by the code. The same argument can be applied to all nodes in $\m{V}(R)$. The linearly independent vectors of the regular sets can be used to reconstruct the data of the source by matrix inversion \footnote{The coding vectors at the edges incoming into each sink can be made known to the sink by the source transmitting the identity matrix. Thus, the sink is informed of
the matrix which it must invert for decoding. This idea is similar to the one used for network coding \cite{chou03}.}.
\end{proof}

As the algorithm proceeds, ports may leave or enter ${\cal Q}_{{\cal C}_t}$. The list ${\cal L}_t$ is then updated accordingly, as we shall discuss in the following sections.

\subsection{Algorithm Description}
The algorithm starts from the upper $R$ input ports of the source $S$ with the standard basis as their coding vectors. The lower input ports of the source have the zero vector as their coding vectors. The source bits are mapped into the source symbols in the field ${\mathbb F}_q$, where $q=2^k$ for some $k$. The transmission is over block length $k=\log|{\mathbb F}_q|$. The vector of source symbols is $\mathbf{\mathbf{u}}=(u_1,\cdots,u_{R})^T$, where $u_i\in {\mathbb F}_q$. The symbol received by a port with coding vector $\mathbf{x} \in {\mathbb F}_q^R$ is $\mathbf{x}^T\mathbf{u}$.

Trivially, the invariant of the algorithm is maintained for the source $S$. At each step of the algorithm, we proceed to the next port in the topological order, determine and update the followings.
\begin{enumerate}
\item The coding coefficients for the new port (and thus the coding vectors).
\item The updated list ${\cal L}_t$ according to the new cut ${\cal C}_t$.
\end{enumerate}
To do so, we shall treat the input and the output ports separately, as discussed in the subsequent sections. For each layer, the coding coefficients for the input ports have to be determined before the coding coefficients for the output ports. We start by considering
coding for the output ports, assuming that the coding vectors at the input ports are given.

\subsubsection{\textbf{Coding for Output Ports}}\label{code_cons_sec}

Assume that at time $t$, the topological order has reached supernode $V$. For the output ports, we assume that in the topological order, $y^V_j$ precedes $y^V_k$ if $j\leq k$. Consider a certain subset $W \in {\cal L}_t$. Some of the ports in this subset can be input ports and some of them output ports, as is the case in Figure \ref{exm3}. This can occur if the topological order has already reached the output ports of several supernodes in the layer, while other output ports at the same layer have not yet been reached. Suppose that the ports in a subset of the list are given by $W=\{w_1,\cdots,w_{R}\}$ and their coding vectors are given by
\begin{equation}
\label{wdef}
\textbf{W}=\{\textbf{w}_1,\cdots,\textbf{w}_{R}\}.
\end{equation}

If the set $W$ contains $p\geq 1$ input ports from $I(V)$, then the subset $W$ has to be updated. After the coding of the output ports of supernode $V$, the input ports in $I(V)$ will be replaced by $p$ of the output ports from $O(V)$. Without loss of generality, assume that the $p$ ports in $W$ that are also in $I(V)$ are $\{w_1,\cdots,w_{p}\}=\{x^V_1,\cdots,x^V_{p}\}$. We choose a set of size $p$ from $O(V)$ and denote the set by $\{w'_1,\cdots,w'_{p}\}$. There are $n \choose p$ such possible sets.

We now update list $\m{L}_t$ to $\m{L}_{t+1}$. In ${\cal L}_{t+1}$, we replace $W = \{w_1, \cdots, w_R\}$ with $W' = \{w'_1, \cdots w'_p, w_{p+1}, \cdots, w_{R}\}$ for all $n \choose p$ choices of  $\{w'_1,\cdots,w'_{p}\}$. In order for the invariant to be maintained, we require the coding vectors of all $n \choose p$ new subsets to be simultaneously linearly independent.

\begin{lemma}
Consider a subset $W \in {\cal L}_{t}$, which contains $p\geq 1$ input ports from $I(V)$. If the field size $q\geq 2$, then there exists a set of coding coefficients $\beta_{(x_i^V,y_j^V)}\in {\mathbb F}_q$ for $1\leq i\leq n,1\leq j\leq n$ such that the coding vectors of a subset $W' \in {\cal L}_{t+1}$ are linearly independent.
\end{lemma}

\begin{proof}
The coding vectors of $W'$ are
\begin{align*}
\mathbf{W}'= \biggl\{ \sum_{i=1}^p \beta_{x_i^V,y_1^V}&\mathbf{x}^V_i+\mathbf{v}_1,\cdots, \sum_{i=1}^p  \beta_{x_i^V,y_p^V}\mathbf{x}^V_i+\mathbf{v}_p,\\
& \mathbf{w}_{p+1},\cdots,\mathbf{w}_{R} \biggr\},
\end{align*}
where $\mathbf{v}_i$, $i=1,\cdots,p$, are the contributions of the coding vectors of the input ports in $I(V) \setminus W$. The $\mathbf{v}_i$'s are assumed to be fixed. Since $W \in{\cal L}_t$, invoking the inductive hypothesis, the set of vectors
$\mathbf{W}=\{\mathbf{x}^V_1,\cdots,\mathbf{x}^V_{p},\mathbf{w}_{p+1},\cdots,\mathbf{w}_{R}\}$
is a basis. We need to determine under which conditions the subset
$\bf{W}'$ is also a basis.

Consider the equation
\begin{align}\label{arel}
\varphi_1&\left(\sum_i \beta_{x_i^V,y_1^V}\mathbf{x}^V_i+\mathbf{v}_1\right)+\cdots+\varphi_p\left(\sum_i \beta_{x_i^V,y_p^V}\mathbf{x}^V_i+\mathbf{v}_p\right)\nonumber\\
&+\varphi_{p+1}\mathbf{w}_{p+1}+\cdots+\varphi_{R}\mathbf{w}_{R}=0.
\end{align}
In order for $\bf{W}'$ to be a basis, we find a sufficient condition for $\varphi_1=\varphi_2 = \cdots=\varphi_{R}=0$ to be the only
solution to (\ref{arel}). First, we express $\mathbf{v}_i$ in the basis $\mathbf{W}$ as
\[
\mathbf{v}_i = \gamma_{1,i}\mathbf{x}^V_1+\cdots+\gamma_{p,i}\mathbf{x}^V_p+\gamma_{p+1,i}\mathbf{w}_{p+1}+\cdots+\gamma_{R,i}\mathbf{w}_{R},
\]
where $\gamma$ are not all zeros. Substituting and rearranging the terms of (\ref{arel}) yields
\begin{align}
&\left[ \varphi_1(\beta_{x^V_1,y^V_1}+\gamma_{1,1})+\cdots+\varphi_p(\beta_{x_1^V,y_p^V}+\gamma_{1,p})\right]\mathbf{x}^V_1+\cdots \nonumber \\
&\hspace*{.2cm}+\left[\varphi_1(\beta_{x^V_p,y^V_1}+\gamma_{p,1})+\cdots+\varphi_p(\beta_{x^V_p,y^V_p}+\gamma_{p,p})\right]\mathbf{x}^V_p\nonumber\\
&\hspace*{.2cm}+\left[\varphi_1\gamma_{p+1,1}+\cdots+\varphi_p\gamma_{p+1,p}+\varphi_{p+1}\right]\mathbf{w}_{p+1}+\cdots\nonumber\\
&\hspace*{.2cm}+\left[\varphi_1\gamma_{R,1}+\cdots+\varphi_p\gamma_{R,p}+\varphi_{R}\right]\mathbf{w}_{R}=0.
\end{align}
Since $\bf{W}$ is a basis, it follows that
\begin{eqnarray}
\varphi_1(\beta_{x^V_1,y^V_1}+\gamma_{1,1})+\cdots+\varphi_p(\beta_{x^V_1,y^V_p}+\gamma_{1,p})&=0\nonumber\\
\vdots& \nonumber\\
\varphi_1(\beta_{x^V_p,y^V_1}+\gamma_{p,1})+\cdots+\varphi_p(\beta_{x^V_p,y^V_p}+\gamma_{p,p})&=0
\end{eqnarray}
This can be written in matrix form as
\begin{eqnarray}
\label{mat_rel}
\begin{pmatrix}
\beta_{x^V_1,y^V_1}+\gamma_{1,1}&\cdots&\beta_{x^V_1,y^V_p}+\gamma_{1,p}\\
\vdots     &\ddots& \vdots\\
\beta_{x^V_p,y^V_1}+\gamma_{p,1}&\cdots&\beta_{x^V_p,y^V_p}+\gamma_{p,p}\\
\end{pmatrix}
\begin{pmatrix}
\varphi_1\\
\vdots \\
\varphi_p
\end{pmatrix}
=\begin{pmatrix}
0\\
\vdots \\
0
\end{pmatrix}.
\end{eqnarray}

We note that $\varphi_1=\cdots=\varphi_p=0$ is the only solution of (\ref{arel}) if and only if the matrix in Equation (\ref{mat_rel}) is non-singular. For a $p\times p$ matrix over a field $\mathbb{F}_q$, the total number of matrices is $q^{p^2}$. Using a combinatorial argument, the number of non-singular matrices is
\begin{align}
\label{lem1eq}
(&q^p-1)(q^p-q)(q^p-q^2)\cdots(q^p-q^{p-1})\\
&=q^{p^2}\left(1-\frac{1}{q^p}\right)\left(1-\frac{1}{q^{p-1}}\right)\cdots\left(1-\frac{1}{q}\right)
 \geq  q^{p^2}(1-\frac{1}{q})^p. \nonumber
\end{align}
Equation (\ref{lem1eq}) can be explained as follows. For the first column of the matrix, we can choose any vector, except the zero vector. There are $q^p-1$ such vectors. For the second column, we can choose any vector except any multiple of the first column (which includes the zero vector). Thus, there are $q^p-q$ choices. In general, there are $q^p-q^{i-1}$ choices for the $i$th column.

So far, we have shown the conditions for $\varphi_1=\cdots=\varphi_p=0$ to be the only solution to (\ref{mat_rel}). If these conditions are maintained, then (\ref{arel}) becomes
\begin{eqnarray}
\varphi_{p+1}\mathbf{w}_{p+1}+\cdots+\varphi_{R}\mathbf{w}_{R}=0.
\end{eqnarray}
The only solution to this relation is $\varphi_{p+1}=\cdots=\varphi_{R}=0$ since the vectors $\mathbf{w}_{p+1},\cdots,\mathbf{w}_{R}$ are in the basis $W$ and are therefore linearly independent.

We conclude that if the matrix in (\ref{mat_rel}) is non-singular, then the vectors in $\mathbf{W'}$ are linearly independent. If $q\ge 2$, then the number of non-singular matrices is positive, and we can choose the set of coding coefficients $\beta_{x^V_i,y^V_j}$, for $1\leq i\leq n,1\leq j\leq n$, such that the matrix is non-singular.
\end{proof}

\begin{lemma}
\label{lem3} If alphabet size $q>n{nN_{layer} \choose R}$, then there exists a set of coding coefficients $\beta_{x^V_i,y^V_j}\in {\mathbb F}_q$ for $1\leq i\leq n,1\leq j\leq n$, such that all the subsets in ${\cal L}_{t+1}$ have linearly independent coding vectors simultaneously.
\end{lemma}

\begin{proof}
The subset $W \in {\cal L}_t$ contains $p\geq 1$ input ports from $I(V)$. From (\ref{lem1eq}), it follows that for a specific subset $W'$ in ${\cal L}_{t+1}$, the number of non-singular matrices is at least
\begin{eqnarray}
q^{p^2}\left(1-\frac{1}{q}\right)^p\geq
q^{p^2}\left(1-\frac{p}{q}\right),
\end{eqnarray}
where the last inequality follows from Bernoulli inequality which holds when $p>0$ and $q\geq 1$. Thus, the number of singular matrices is at most
\begin{eqnarray}
q^{p^2}-q^{p^2}\left(1-\frac{p}{q}\right)=pq^{p^2-1}\leq
nq^{n^2-1}.
\end{eqnarray}
In ${\cal L}_{t+1}$, there are at most $nN_{layer} \choose R$ newly added subsets. For each subset, there are at most $nq^{n^2-1}$ choices of a set of coding coefficients $\beta_{x^V_i,y^V_j}$, for $1\leq i\leq n,1\leq j\leq n$, such that the coding vector associated with the subset is linearly dependent. Therefore by the union bound there are at most ${nN_{layer} \choose R}nq^{n^2-1}$ choices of sets of coding coefficients such that at least one of the subsets in ${\cal L}_{t+1}$ can have dependent coding vectors. The total number of choices of coding coefficients is $q^{n^2}$. Therefore, if $q>n{nN_{layer} \choose R}$, then we will have at least a single set of coding coefficients such that all the subsets in ${\cal L}_{t+1}$ have linearly independent coding vectors simultaneously.
\end{proof}

We note that for each supernode, the coding vectors of the output ports can be viewed as columns of a parity check matrix of a Maximum Distance Separable (MDS) code with parameters $(n,k=p)$.

\begin{theorem}
The invariant of the algorithm is maintained for the output ports.
\end{theorem}

\begin{proof}
By assumption, the invariant is maintained for the set ${\cal Q}_{\m{C}_{t(I)}}$, which contains the input ports of the supernodes in layer $l$. We need to show that the invariant is maintained for the set ${\cal Q}_{t(O)}$, which contains the output ports of the supernodes in layer $l$, where $t(I)<t(O)$. This follows by induction from Lemma \ref{lem3}.
\end{proof}

The average complexity of this stage is computed using arguments similar to those in \cite{jaggi05} for the network code construction. According to Lemma \ref{lem3}, we can choose the alphabet size at this stage to be $q=2n{nN_{layer}  \choose R}$. It follows from the proof of Lemma \ref{lem3} that the probability of failure when the coding coefficients are chosen randomly is upper bounded by
\begin{equation}
P_{f}\leq\frac{{nN_{layer} \choose  R}nq^{n^2-1}}{q^{n^2}}=\frac{{nN_{layer} \choose R}n}{q}=\frac{1}{2}.
\end{equation}
Therefore, the expected number of trials until the vectors in $\mathbf{W'}$ form a basis is at most $2$. A single layer has at most $N_{layer}$ supernodes. The total number of edges connecting input and output ports of a certain supernode is $n^2$. It follows that the total number of edges at the layer is bounded by $N_{layer}n^2$. In our case, the equivalent to the number of sinks $|\m{T}|$ in \cite{jaggi05} is the size of ${\cal L}_t$, which is at most ${nN_{layer} \choose R}$. Therefore, similarly to the complexity in \cite{jaggi05} for network coding, the average complexity for a layer is $O({nN_{layer} \choose R}N_{layer}n^2 R)$. If the total number of layers is $\lambda$, then the total average complexity of finding the coding coefficients of the output ports is $O({nN_{layer} \choose R}N_{layer}n^2R\lambda)$.

\subsubsection{\textbf{Coding for Input Ports}}\label{code_cons_input}
The coding for the input ports is performed jointly over all supernodes in the same layer. Assume that the coding coefficients
of the output ports of layer $l$ have already all been updated according to Section \ref{code_cons_sec}. We need to update the
coding coefficients of the input ports of layer $l+1$. The list ${\cal L}_t$ contains ports from layer $l$ only. From the list $\m{L}_t$,
we choose an arbitrary subset
\begin{equation}
W=\{w_1,\cdots,w_{R}\} \in \m{L}_t.
\end{equation}
The set $W'$ is a subset of $R$ input ports at layer $l+1$. Consider the bipartite network that consists of $W$ and the input
ports of layer $l+1$, with edges from ports in $W$ to ports in layer $l+1$. Let $\m{B}(W, W')$ denote the bipartite subgraph with vertices $(W, W')$ and edges from ports in $W$ to ports in $W'$. Let $G_{(W, W')}$ denote the incidence matrix for ${\cal B}(W,W')$. If $G_{(W,W')}$ is full rank, then we shall show that we can find coding coefficients such that the coding vectors of the ports in $W'$ are linearly independent. If we cannot find a set $W'$ such that $G_{(W, W')}$ is full rank, then we remove $W$ from the list ${\cal L}_{t}$ and do not replace it with a new set $W'$. Nevertheless, we shall show that whenever $W$ is regular, we can always find a regular $W'$ such that $G_{(W,W')}$ is full rank.

In Figure \ref{exm4}, we see the sets $W,W'$. It can be verified that $R=3$; however, the rank of the incidence matrix, $\text{rank}(G_{(W, W')})=2$. The set $W'$ is not regular according to Definition \ref{den_regular} since the upper and the lower ports in $W'$ always receive the same symbol.

\begin{figure}[tbp]
  \begin{center}
    \leavevmode
    \epsfig{file=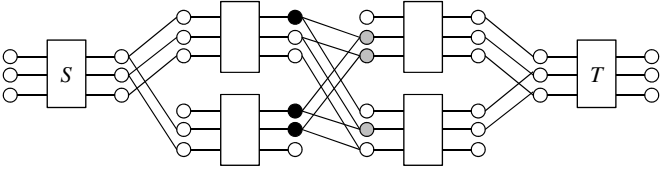,width=8.8cm}
        \caption{Example of a non-regular $W'$. Ports in black are in $W$, and ports in grey are in $W'$.}
    \label{exm4}
  \end{center}
\end{figure}

\begin{lemma}
\label{conv} If $W^{\prime }$ is a regular set containing $R$ input ports of supernodes at layer $l+1$ for some $1\leq l\leq
\lambda -1$, then there exists a regular set $W$ containing $R$ output ports on supernodes at layer $l$ such that the incidence matrix $G_{(W,W')}$ is full rank.
\end{lemma}

\begin{proof}
The result follows from \cite{sadegh09}\cite{goemans}\cite{amaudruz09}\footnote{These works consider the problem of unicast communication. We can use their result for the proof of Lemma \ref{conv} since we are interested here in a single set $W'$ with a corresponding virtual sink $T(W')$.}. Specifically, in \cite{amaudruz09}, a path is defined as a disjoint set of edges $(e_1,\cdots,e_{\mu})$ where $e_1$ starts from the source, $e_{\mu}$ enters a certain sink, and $e_i$ enters the same supernode from which $e_{i+1}$ emerges. In this proof, we consider the virtual sink $T(W')$ as our sink.

In \cite{amaudruz09}, linearly independent (LI) paths are defined. Consider the subgraph $G'$ of the network $G$ consisting of $K$ paths from source $S$ to $T(W')$. The paths are LI if in $G'$ the rank of the incidence matrix of any cut is exactly $K$. We call a set of $R$ LI paths the \textit{underlying flow} $F_u$. It has been shown in \cite{amaudruz09} that such an underlying flow $F_u$ exists. In $F_u$, we can consider the $R$ output ports of layer $l$, one from each path. This set of $R$ ports is guaranteed to be regular, by definition of the underlying flow. This set of output ports will be chosen as $W$. The $\text{rank}(G_{(W, W')})$ is full rank, again by definition of linearly independent paths. Therefore, the two properties of the lemma are maintained.
\end{proof}

We note that in our construction we do not need to find the edges in $F_u$. The concept of the underlying flow was introduced only for
the proof of Lemma \ref{conv}.

\begin{lemma}
\label{dir} For a regular set $W'$, we can find coding coefficients such that the coding vectors $\mathbf{W'}$ are linearly independent simultaneously if the alphabet size $q>2{nN_{layer}\choose R}$.
\end{lemma}

\begin{proof}
By Lemma \ref{conv}, given a regular set $W'$, there exists a regular set $W$ containing $R$ output ports of supernodes in layer $l$ such that the incidence matrix $G_{(W,W')}$ is full rank. The coding vectors of the output ports in $W$ are given by
\begin{equation}
\mathbf{W}=\{\mathbf{w}_1,\cdots,\mathbf{w}_{R}\}.
\end{equation}
The coding vectors in the input ports in $W'$ are given by
\begin{equation}
\mathbf{W'}=\{\mathbf{w}'_1,\cdots,\mathbf{w}'_{R}\}.
\end{equation}
The vector $\mathbf{w}_i$ is in the form
\begin{equation}
\label{vieq}
\mathbf{w}'_i=\sum_{j=1}^{R}g_{i,j}k_{j}\mathbf{w}_j+\mathbf{\tilde{w}_i}
\end{equation}
where $k_{j}$ are the virtual coefficients, and $\mathbf{\tilde{w}}_i$ is the contribution of output ports of layer $l$ that are not in $W$. The binary coefficient $g_{i,j}$ is the $(i,j)$th element of matrix $G_{(W,W')}$.

We need to find the conditions on the coefficient $k_j$ under which the ports in $W'$ have coding vectors which are linearly
independent. Consider the equation
\begin{equation}
\varphi_1\mathbf{w}'_1+\cdots+\varphi_{R}\mathbf{w}'_{R}=0.
\end{equation}
Combining with (\ref{vieq}), and rearranging,
\begin{align}
\label{via} &\varphi_1\left(\sum_{j=1}^R g_{1,j}k_{j}\mathbf{w}_j\right)+\cdots+\varphi_{R}\left(\sum_{j=1}^R g_{R,j}k_{j}\mathbf{w}_j\right)\nonumber\\
&=-\alpha_1\mathbf{\tilde{w}_1}-\cdots-\alpha_{h_{min}}\mathbf{\tilde{w}_{h_{min}}}.
\end{align}
We can represent vector $\mathbf{\tilde{w}}_i$ in the basis $\mathbf{W}$ as
\begin{equation}
\mathbf{\tilde{w}_i}=\gamma_{1,i}\mathbf{w}_1+\cdots+\gamma_{R,i}\mathbf{w}_{R}.
\end{equation}
Combining with (\ref{via}) and rearranging terms,
\begin{align}
\label{via2}
(\varphi_1 \gamma_{1,1}+&\cdots+\varphi_{R}\gamma_{1,R}+\varphi_1 g_{1,1}k_{1} +\cdots+ \varphi_{R}g_{R,1}k_{1})\mathbf{w}_1\nonumber\\
+\cdots&+ (\varphi_1\gamma_{R,1}+\cdots+\varphi_{R}\gamma_{R,R}\nonumber\\
&+\varphi_1h_{1,R}k_{R}+\cdots+\varphi_{R}g_{R,R}k_{R})\mathbf{w}_{R}=0.
\end{align}
Since $W$ is a basis, it follows that
\begin{eqnarray}
\varphi_1\gamma_{1,1}+\cdots+\varphi_{R}\gamma_{1,R}+\gamma_1g_{1,1}k_{1}+\cdots+\varphi_{R}g_{R,1}k_{1}&=0\nonumber\\
\vdots &\nonumber\\
\varphi_1\gamma_{R,1}+\cdots+\varphi_{R}\gamma_{R,R}\hspace*{4.3cm}& \nonumber \\
+\varphi_1 g_{1,R}k_{R}+\cdots+\varphi_{R}g_{R,R,}k_{R}&=0,\nonumber
\end{eqnarray}
or in matrix notation,
\begin{small}
\begin{eqnarray}
\label{mat_rel5}
\begin{pmatrix}
\gamma_{1,1}+g_{1,1}k_{1}&\cdots&\gamma_{1,R}+g_{R,1}k_{1}\\
\vdots     &\ddots& \vdots\\
\gamma_{R,1}+g_{1,R}k_{R}&\cdots&\gamma_{R,R}+g_{R,R}k_{R}\\
\end{pmatrix}\begin{pmatrix}
\varphi_1\\
\vdots \\
\varphi_{R}
\end{pmatrix}=\begin{pmatrix}
0\\
\vdots \\
0
\end{pmatrix}
\end{eqnarray}
\end{small}

Let $H$ denote the matrix on the left-hand side of (\ref{mat_rel5}). The zero vector is the only solution to (\ref{mat_rel5}) if and only if the matrix $H$ is full rank. The determinant of the matrix $H$ is a polynomial in the parameters $\{\gamma_{i,j},k_{j},g_{i,j}$, for $1\leq
i,j\leq R\}$. Denote the polynomial by $\Delta_{W,W'}(\beta_{i,j},k_{j},h_{i,j})$. When all $\gamma_{i,j}=0$, row $i$ of $H$ is equal to the $i$th row of $G_{(W, W')}$ multiplied by $k_i$. Therefore, the polynomial $\Delta_{W,W'}$ is of the form:
\begin{align}
\label{delww} \Delta_{W,W'}(\gamma_{i,j},k_{j},g_{i,j})=\text{det}(G_{(W, W')})\prod_{j=1}^R k_j+\delta(\gamma_{i,j},k_{j},g_{i,j})
\end{align}
where $\text{det}(G_{(W, W')})\neq 0$ is the determinant of (non-singular) matrix $G_{W,W'}$, and
$\delta(\cdot)$ is a polynomial such that the sum of the degrees of all the parameters $k_{j},1\leq j\leq R$, is smaller than
$R$. It follows that for constant $\gamma_{i,j},g_{i,j}$, $1\leq i,j\leq R$, $\Delta_{W, W'}$ is not the zero polynomial.

The polynomial $\Delta_{W,W'}$ corresponds to the pair of regular subsets $(W,W')$. We need to find the corresponding polynomials
for all regular $W'$. Let $\m{P}_t$ denote the set of all these pairs $(W,W')$. In order for all such sets $W'$ to have independent coding vectors, we need to assign the coding coefficients such that the following polynomial does not vanish to zero.
\begin{equation}
\label{delw}P=\prod_{W':\exists(W,W')\in {\cal P}_t} \Delta_{W,W'}.
\end{equation}
By definition, $P$ is not the zero polynomial since it is a product of nonzero polynomials. Hence, there is a set of coefficients
$k_{j}$ such that the polynomial does not vanish to zero.

Next, we discuss how such coefficients $k_{j}$'s can be found. Reference \cite{algebraic} proposes an algorithm to find a vector $\underline{a}$ such that a given polynomial $P$ evaluated at $\underline{a}$ is not equal to zero, \ie $P(\underline{a}) \ne 0$. We reproduce this algorithm from \cite{algebraic} in Algorithm \ref{alg:poly} for completeness.

\begin{algorithm}[h!]
\SetKwInOut{innnn}{Input}\SetKwInOut{Output}{Output}
\SetKwInOut{one}{Step 1}\SetKwInOut{two}{Step 2}\SetKwInOut{three}{Step 3}
\innnn{A polynomial $P$ in variables $\xi_1, \xi_2, \cdots, \xi_n$; integers $i = 1$, $t=1$}
\Output{$\underline{a} = (a_1, a_2, ..., a_n)$}
\one{Find the maximal degree $d$ of $P$ in any variable $\xi_i$, and let $i$ be the smallest number such that $2^i > d$\;}
\two{Find an element $a_t$ in $\mathbb{F}_{2^i}$ such that $P(\xi)|_{\xi_t = a_t} \ne 0$ and let $P \leftarrow P(\xi)|_{\xi_t = a_t}$\;}
\three {If {$t = n$, then halt; otherwise, $t \leftarrow t+1$ and go to \textbf{Step 2}\;}}\vspace*{.3cm}
\caption{Algorithm to find $\underline{a}$ such that $P(\underline{a}) \ne 0$ \cite{algebraic}.}\label{alg:poly}
\end{algorithm}

In our scenario, the maximal degree of each variable in $\Delta_{W,W'}$ is $1$ because of the structure of the matrix. It
follows that the maximal degree of each variable in $P$ is at most ${nN_{layer} \choose R}$. Therefore, $d={nN_{layer} \choose
R}$ and we can always choose $i=\lceil\log{nN_{layer} \choose R}\rceil$. It follows that an alphabet larger than $2^{\lceil\log{nN_{layer} \choose R}\rceil}\leq 2{nN_{layer} \choose R}$ will ensure that there exists coding coefficients such that $\mathbf{W'}$ are linearly independent.
\end{proof}

\begin{theorem}
The invariant of the algorithm is maintained for input ports.
\end{theorem}

\begin{proof}
We prove the theorem by induction. For the base case, consider the $R$ upper output ports of the source $S$. We assign the standard basis as coding vectors for these $R$ output ports. We then apply Lemmas \ref{conv} and \ref{dir} to the first layer.

For the inductive step, assume that the statement holds for $\m{Q}_t$, where $\m{Q}_t$ contains the output ports of the supernodes in layer $l$. Now, we show that the invariant is maintained for ${\cal Q}_{t+1}$, which contains the input ports of the supernodes in layer $l+1$. If $W'$ is a regular subset, then by Lemma \ref{dir}, there is a coding assignment (according to the code construction presented) such that vectors in $\mathbf{W'}$  are independent. Therefore, the invariant is maintained also for layer $l+1$.
\end{proof}

We now analyze the complexity. For each pair of sets $(W,W')$, we need to verify whether the incidence matrix $G_{(W,W')}$ of the bipartite graph ${\cal B}(W,W')$ is full rank. This can be performed by determinant computation in complexity $O(R^3)$. Therefore, the total complexity of this stage for a single layer is $O(R^3{nN_{layer} \choose R})$.

For each variable of $P$, we require an iteration of the Algorithm \ref{alg:poly}. Each iteration of Algorithm \ref{alg:poly} takes at most $2^{\lceil\log{n N_{layer} \choose R}\rceil}$ assignments need to be verified. There are at most $N_{layer}$ supernodes at each layer. Therefore, the maximal number of output ports, which is also the number of variables $k_j$, is $nN_{layer}$. It follows that the complexity of the coding for input ports for a single layer is $O({nN_{layer} \choose R}(nN_{layer}+R^3)$. Therefore, if we consider all $\lambda$ layers, the complexity for the coding for the input ports is $O({nN_{layer} \choose R}\lambda(nN_{layer}+R^3))$. Combining the complexity for both input and output ports, it follows that the total complexity of the algorithm is $O({nN_{layer} \choose R}N_{layer}n^2 \lambda R)$. 

\section{Conclusions}\label{sec:conclusions}

%
%
ADT networks \cite{adt1}\cite{adt2} have drawn considerable attention for their potential to approximate the capacity of wireless relay networks. In this paper, we showed that the ADT network can be described well within the algebraic network coding framework \cite{algebraic}. This connection between ADT network and algebraic network coding allows the use of results on network coding to understand better the ADT networks. 

In this paper, we derived an algebraic definition of min-cut for the ADT networks, and provided an algebraic interpretation of the Min-cut Max-flow theorem for a single unicast/mulciast connection in ADT networks. Furthermore, by taking advantage of the algebraic structure, we have shown feasibility conditions for a variety of set of connections $\m{C}$, such as multiple multicast, disjoint multicast, and two-level multicast. We also showed optimality of linear operations for the connections listed above in the ADT networks, and showed that random linear network coding achieves the capacity. Furthermore, we extended the capacity characterization to networks with cycles and random erasures/failures. We proved the optimality of linear operations (as well as random linear network coding) for multicast connections in ADT networks with cycles. By incorporating random erasures into the ADT network model, we showed that random linear network coding is robust against failures and erasures.

Taking advantage of this insight, we proposed an efficient linear code construction for multicasting in ADT networks while guaranteeing decodability. The average complexity of the construction is $O({nN_{layer} \choose R}N_{layer} n^2 \lambda R)$. The required field size is at most $q=n {nN_{layer} \choose R}$; thus, the block length required to represent a symbol is at most $k=\log(n{N_{layer}n \choose
R})$. Our code construction does not require finding network flows or knowing the exact location of the sinks. When normalized by the number of sinks, our code construction has a complexity which is comparable to those of previous coding schemes for a single sink. A possible direction for future research is to use our construction to find new coding schemes for practical multiuser networks with receiver noise.

\bibliographystyle{IEEEtran}
\bibliography{References}

\end{document}